%% file: main.tex
\RequirePackage{rotating}
\documentclass[sigconf,screen]{acmart}
\AtBeginDocument{%
  }

\copyrightyear{2025}
\acmYear{2025}
\setcopyright{cc}
\setcctype{by}
\acmConference[SPAA '25]{37th ACM Symposium on Parallelism in Algorithms and Architectures}{July 28-August 1, 2025}{Portland, OR, USA}
\acmBooktitle{37th ACM Symposium on Parallelism in Algorithms and Architectures (SPAA '25), July 28-August 1, 2025, Portland, OR, USA}
\acmDOI{10.1145/3694906.3743311}
\acmISBN{979-8-4007-1258-6/2025/07}

\usepackage{booktabs}   
\usepackage{subcaption} 
\usepackage{caption}
\usepackage{colortbl}
\usepackage{bigstrut}
\usepackage{microtype}
\usepackage{tabularx}
\usepackage{multirow}
\usepackage{multicol}
\usepackage{rotating}
\usepackage{lipsum}
\usepackage{balance}
\usepackage{mfirstuc}
\usepackage{titlecaps}
\usepackage{listings}
\usepackage{float}
\usepackage[rightcaption]{sidecap}

\usepackage{xcolor}
\usepackage[font=small,labelfont=bf,skip=2pt]{caption}

\input{sty.tex}
\input{macro.tex}

\begin{document}

\makeatletter
\gdef\@copyrightpermission{
  \begin{minipage}{0.2\columnwidth}
   \href{https://creativecommons.org/licenses/by/4.0/}{\includegraphics[width=0.90\textwidth]{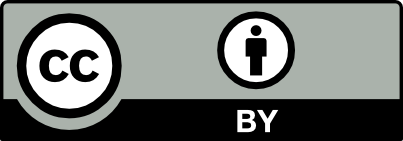}}
  \end{minipage}\hfill
  \begin{minipage}{0.8\columnwidth}
   \href{https://creativecommons.org/licenses/by/4.0/}{This work is licensed under a Creative Commons Attribution International 4.0 License.}
  \end{minipage}
  \vspace{5pt}
}
\makeatother

\title{Parallel Point-to-Point Shortest Paths and Batch Queries}
\settopmatter{authorsperrow=4}

\author{Xiaojun Dong}
\affiliation{%
  \institution{UC Riverside}
  \city{Riverside}
  \state{CA}
  \country{USA}
}
\email{xdong038@ucr.edu}

\author{Andy Li}
\affiliation{%
  \institution{UC Riverside}
  \city{Riverside}
  \state{CA}
  \country{USA}
}
\email{ali164@ucr.edu}

\author{Yan Gu}
\affiliation{%
  \institution{UC Riverside}
  \city{Riverside}
  \state{CA}
  \country{USA}
}
\email{ygu@cs.ucr.edu}

\author{Yihan Sun}
\affiliation{%
  \institution{UC Riverside}
  \city{Riverside}
  \state{CA}
  \country{USA}
}
\email{yihans@cs.ucr.edu}

\input{abstract}



\renewcommand\footnotetextcopyrightpermission[1]{} 
\fancyhead{} 






\maketitle

\input{intro}

\input{prelim}
\input{algo}

\input{framework}
\input{imp}

\input{exp}
\input{related-work}

\input{concl}

\begin{acks}
This work is supported by NSF grants CCF-2103483, IIS-2227669, and TI-2346223, NSF CAREER Awards CCF-2238358 and CCF-2339310, the UCR Regents Faculty Development Award, and the Google Research Scholar Program.
We thank the anonymous reviewers for their useful comments.
\end{acks}

\input{ppsp.bbl}

\iffullversion{
\appendix
\input{appendix}

}

\end{document}
\endinput

%% file: sty.tex
\def\fullversion
\usepackage{etoolbox}
\newcommand{\ifconference}[1]{{{\ifx\fullversion\undefined{#1}\fi}\xspace}}
\newcommand{\iffullversion}[1]{{{\ifx\conference\undefined{#1}\fi}\xspace}}

\usepackage{graphicx}  
\usepackage{lipsum}  
\newcommand{\hide}[1]{} 
\usepackage{xspace}
\usepackage{textcomp}
\usepackage{comment} 
\usepackage{verbatim}

\ifx\nocomment\undefined
\newcommand{\yan}[1]{{\textcolor{violet}{Yan: #1}}}
\newcommand{\yihan}[1]{{\color{red}{\bf Yihan:} #1}}
\newcommand{\xiaojun}[1]{{\color{cyan}{\bf Xiaojun:} #1}}
\newcommand{\andy}[1]{{\textcolor{purple}{andy: #1}}}
\newcommand{\revise}[1]{{#1}}
 
\else
\newcommand{\yan}[1]{}
\newcommand{\yihan}[1]{}
\newcommand{\xiaojun}[1]{}
\newcommand{\andy}[1]{}
\newcommand{\revise}[1]{{#1}}
\fi

\usepackage{url}

\newcommand{\mtextsc}[1]{{\mbox{\sc{#1}}}} 
\newcommand{\defn}[1]{\emph{\textbf{#1}}} 
\newcommand{\emp}[1]{\emph{\boldmath\textbf{#1}\unboldmath}} 

\usepackage[ruled,lined,linesnumbered,noend]{algorithm2e}
\usepackage[noend]{algpseudocode}

\makeatletter
\patchcmd{\@algocf@start}
  {-1.5em}
  {0pt}
  {}{}
\newcommand{\nosemic}{\renewcommand{\@endalgocfline}{\relax}}
\newcommand{\dosemic}{\renewcommand{\@endalgocfline}{\algocf@endline}}

\SetSideCommentLeft
\SetKwInput{notations}{Notations}
\SetKwInput{notes}{Notes}
\SetKwInput{maintains}{Maintains}

\SetKwProg{myfunc}{Function}{}{}
\SetKwFor{parForEach}{ParallelForEach}{do}{endfor}
\SetKwFor{Justrepeat}{Repeat}{}{}

\SetKw{MIN}{min}
\SetKw{MAX}{max}
\SetKw{OR}{or}
\SetKw{AND}{and}

\SetCommentSty{mycommfont}

\usepackage{cleveref}
\crefname{section}{Sec.}{Sec.}
\crefname{theorem}{Thm.}{Thm.}
\crefname{lemma}{Lem.}{Lem.}
\crefname{corollary}{Col.}{Col.}
\crefname{table}{Tab.}{Tab.}
\crefname{algorithm}{Alg.}{Alg.}
\crefname{figure}{Fig.}{Fig.}
\crefname{fact}{Fact}{Fact}
\Crefname{table}{Tab.}{Tab.}
\crefname{problem}{Problem}{Problem}

\newtheorem{theorem}{Theorem}[section]

\newtheorem{fact}[theorem]{Fact}

\let \originalleft \left
\let\originalright\right
\renewcommand{\left}{\mathopen{}\mathclose\bgroup\originalleft}
\renewcommand{\right}{\aftergroup\egroup\originalright}

\usepackage{scalerel} 

\newtheoremstyle{exampstyle}
{.5em} 
{1em} 
{\it} 
{.5em} 
{\it \bfseries} 
{.} 
{.5em} 
{} 
\theoremstyle{exampstyle} 
\theoremstyle{exampstyle} 
\theoremstyle{exampstyle} 
\theoremstyle{exampstyle} 

\makeatletter
\renewenvironment{proof}[1][\proofname]{\par
\vspace{-2\topsep}
\pushQED{\qed}%
\normalfont
\topsep0pt \partopsep0pt 
\trivlist
\item[\hskip\labelsep
      \itshape
  #1\@addpunct{.}]\ignorespaces
}{%
\popQED\endtrivlist\@endpefalse
}

\newcommand{\R}{\mathbb{R}}

\newcommand{\whp}[1]{\emph{whp}}

\newcommand{\true}{\emph{true}}
\newcommand{\false}{\emph{false}}

\usepackage{pifont}

\newcommand{\modelop}[1]{\texttt{#1}}
\newcommand{\forkins}{\modelop{fork}}

\newcommand{\thread}{thread}

\newcommand{\insformat}[1]{\texttt{#1}}

\newcommand{\WriteMin}{\insformat{write\_min}\xspace}

\newcommand{\writemin}{\WriteMin}

\usepackage[shortlabels]{enumitem}

\setlist{topsep=0.3em,itemsep=0.2em,parsep=0.1em,leftmargin=*}

\usepackage{float}
\usepackage[labelfont=bf,font={small},aboveskip=0em, belowskip=0em]{caption}

\usepackage[labelfont=bf,list=true,skip=0em]{subcaption}
\usepackage[rightcaption]{sidecap}

\usepackage{wrapfig}

\usepackage{array}
\newcolumntype{L}[1]{>{\raggedright\let\newline\\\arraybackslash\hspace{0pt}}m{#1}}
\newcolumntype{C}[1]{>{\centering\let\newline\\\arraybackslash\hspace{0pt}}m{#1}}
\newcolumntype{R}[1]{>{\raggedleft\let\newline\\\arraybackslash\hspace{0pt}}m{#1}}
\newcolumntype{B}{>{\bf}c}
\usepackage{rotating}

\usepackage{booktabs} 
\usepackage{multicol,multirow}
\usepackage{longtable} 
\usepackage{supertabular} 
\usepackage{colortbl}
\usepackage{bigstrut}

\newcommand{\myparagraph}[1]{\vspace{.1em}\noindent\emp{#1}\enspace}

\usepackage{mdframed}
\definecolor{framelinecolor}{RGB}{68,114,196}
\mdfdefinestyle{mystyle}{linecolor=framelinecolor,innertopmargin=1pt,innerbottommargin=2pt,backgroundcolor=gray!20,skipabove=2pt,skipbelow=0pt}
\mdfdefinestyle{densestyle}{linecolor=framelinecolor,innertopmargin=0,innerbottommargin=0,leftmargin=0,rightmargin=0,backgroundcolor=gray!20}
\mdfdefinestyle{compactcode}{linecolor=framelinecolor,innertopmargin=1pt,innerbottommargin=1pt,backgroundcolor=gray!20,skipabove=0pt,skipbelow=0pt,leftmargin=0,rightmargin=0}

\usepackage{framed}

\usepackage{listings}

\newdimen\zzsize
\newdimen\kwsize
\newcommand{\basicstyle}{\fontsize{\zzsize}{1\zzsize}\ttfamily}
\newcommand{\keywordstyle}{\fontsize{\kwsize}{1\kwsize}\ttfamily\bf}

\newdimen\zzlstwidth
\usepackage{tikz} 

\hide{
\binoppenalty=700
\brokenpenalty=0 
\clubpenalty=0   
\displaywidowpenalty=0   
\exhyphenpenalty=50
\floatingpenalty=20000
\hyphenpenalty=50
\interlinepenalty=0
\linepenalty=10
\postdisplaypenalty=0
\predisplaypenalty=0 
\relpenalty=500
\widowpenalty=0  
}

\usepackage{titlecaps} 

%% file: macro.tex
\newcommand{\lazyinsert}{\mtextsc{Add}\xspace}
\newcommand{\extract}{\mtextsc{Extract}\xspace}
\newcommand{\extcond}{\mtextsc{GetDist}\xspace}

\newcommand{\ourppsp}{our PPSP} 

\newcommand{\bds}{BiDS}
\newcommand{\bids}{\bds}
\newcommand{\bdastar}{BiD-A$^*$}
\newcommand{\ourlib}{Orionet} 

\newcommand{\querygraph}{query graph} 
\newcommand{\astar}{A$^*$}
\newcommand{\gq}{G_q}
\newcommand{\vq}{V_q}
\newcommand{\eq}{E_q}
\newcommand{\nq}{N_q}
\newcommand{\mumax}{\mu_{\max}}
\newcommand{\vcopy}[2]{#1^{\langle #2 \rangle}}

\newcommand{\naive}{na\"ive} 

\newcommand{\initial}{\mtextsc{Init}}
\newcommand{\prune}{\mtextsc{Prune}}
\newcommand{\update}{\mtextsc{UpdateDistance}}

\newcommand{\tbf}[1]{{\boldmath \textbf{#1} \unboldmath}}

\newcommand{\knn}{$k$-NN}

%% file: abstract.tex
\begin{abstract}
We propose Orionet, efficient parallel implementations of Point-to-Point Shortest Paths (PPSP) queries using bidirectional search (BiDS) and other heuristics, with an additional focus on batch PPSP queries.
We present a framework for parallel PPSP built on existing single-source shortest paths (SSSP) frameworks
by incorporating novel pruning conditions.
As a result, we develop efficient parallel PPSP algorithms based on early termination, bidirectional search, A$^*$ search,
and bidirectional A$^*$---all with simple and efficient implementations.

We extend our idea to batch PPSP queries, which are widely used in real-world scenarios.
We first design a simple and flexible abstraction to represent the batch
so PPSP can leverage the shared information of the batch.
Orionet formalizes the batch as a \emph{query graph} represented by edges between queried sources and targets.
In this way, we directly extended our PPSP framework to batched queries in a simple and efficient way. 

We evaluate Orionet on both single and batch PPSP queries using various graph types and distance percentiles of queried pairs, 
and compare it against two baselines, GraphIt and MBQ. 
Both of them support parallel single PPSP and \astar{} using unidirectional search.
On 14 graphs we tested, on average, our bidirectional search is 2.9$\times$ faster than GraphIt, and 6.8$\times$ faster than MBQ. 
Our bidirectional \astar{} is 4.4$\times$ and 6.2$\times$ faster than the \astar{} in GraphIt and MBQ, respectively.
For batched PPSP queries, we also provide in-depth experimental evaluation,
and show that Orionet provides strong performance compared to the plain solutions.  

\hide{
Compared to full SSSP, our BiDS-based single PPSP achieves 1.21-101$\times$ speedups
for sources and targets at close or moderate distance.
It can be particularly efficient (more than 100$\times$ faster) when the queried pair is nearby.
Orionet also provides the first parallel solution for general batch PPSP queries.
On several carefully designed benchmarks that cover different query types of batch PPSP,
the new solutions in Orionet achieved 1.92$\times$ average speedup
than a na\"ive SSSP baseline.}
\end{abstract}

%% file: intro.tex
\section{Introduction}\label{sec:intro}

In this paper, we propose \ourlib{}, an efficient parallel library for Point-to-Point Shortest Paths (PPSP) queries using bidirectional search and other heuristics,
with the additional support of \emph{batch PPSP queries}.
PPSP is an important graph processing problem.
Given a graph $G=(V,E)$ and two vertices $s,t\in V$, 
a PPSP query finds the shortest paths from $s$ to $t$.
For simplicity in description, we assume the graphs are undirected in this paper, but all techniques in this paper also apply to directed graphs.
It is one of the most widely used primitives on graphs with applications such as navigation, network design, and artificial intelligence~\cite{luxen2011real,GraphHopper,magnanti1993shortest,mitchell2000geometric,kumawat2021extensive}.

Although in theory, PPSP has the same asymptotic cost as the single-source shortest paths (SSSP) problem (i.e., finding the shortest paths from $s$ to all $v\in V$), 
in practice, many existing techniques can accelerate PPSP.
One of the most widely used techniques is the \defn{bidirectional search} (\bds{}).
At a high level, a \bds{} runs SSSP (e.g., Dijkstra's algorithm) from both $s$ and $t$, with a
stop (pruning) condition, such that a PPSP query does not need to traverse the entire graph. 
Hence, the cost of \bds{} can be substantially lower than a complete SSSP query from $s$ (see \cref{fig:ppsp} as an illustration).
On graphs with geometric information, other heuristics, such as \astar{}, can also be used, either independently or together with \bds{}.

Sequentially, \bds{} has been widely studied and shown to be effective in accelerating PPSP queries.
However, despite related work on parallel PPSP~\cite{barley2018gbfhs,geisberger2008contraction,dibbelt2016customizable,xu2019pnp,zhang2020optimizing}, we are unaware of any work \emph{parallelizing the \bds{}}.
One major reason for this gap is the inherent distinction in the approach of parallel SSSP compared to its sequential counterpart.
Sequentially, the efficiency of \bds{} comes from the fact that SSSP queries from both ends visit vertices in \emph{increasing distance order (i.e., Dijkstra's ordering)}.
The standard stop condition for \bds{} is when the first vertex $v$ is settled (finalized) from both directions.
Hence, \bds{} saves much work by skipping vertices, such as those that are further than $v$ from $s$ (and similarly for the backward search). 
However, to achieve high parallelism,
practical parallel SSSP algorithms usually do not visit vertices in ascending order of distance (consider parallel Bellman-Ford).
Therefore, it remains unclear \emph{whether and how the bidirectional search, shown effective for sequential PPSP queries,
can be parallelized \textbf{correctly} and \textbf{efficiently}.} 

Designing parallel \astar{} is relatively straightforward.
\astar{} incorporates a heuristic function from the current vertex to the target $t$, adding it to the tentative distance from $s$ to prune the search in the ``wrong'' direction.
Many existing parallel SSSP algorithms can be adapted for \astar{}~\cite{zhang2020optimizing,dong2021efficient,zhang2024multi}.
Interestingly, \bds{} and \astar{} can be combined as bidirectional \astar{} (\bdastar{}).
However, all existing sequential \bdastar{} algorithms~\cite{goldberg2005computing,ikeda1994fast} rely on strictly increasing distance orders, similar to \bds{}.
Thus, efficiently parallelizing \bdastar{} also remains an open challenge.

The first contribution of \ourlib{} is \emp{a framework of algorithms and implementations for efficient parallel PPSP, using techniques including bidirectional search, \astar{} search, and bidirectional \astar{} search.}
To achieve high performance and low coding effort,
we adopt a parallel SSSP framework called \emph{stepping algorithms}~\cite{dong2021efficient}.
We carefully redesigned the framework for PPSP queries, which abstract three user-defined functions for different approaches:
1) the setup of the initial frontier (vertices to be processed in the first round),
2) the prune condition to skip certain vertices, and 
3) the process to update the current best answer.
Interestingly, this abstraction also extends to batch PPSP queries, which we describe in detail below.
Our key algorithmic insight mainly lies in the new prune conditions, which correctly and effectively reduce the searching space,
and are also simple for implementation in parallel. 
We provide the technical details in \cref{sec:algo}.
Experimentally, we show that our new approaches in \ourlib{} based on \bds{}, \astar{}, and \bdastar{} greatly outperform existing baselines.

\input{figures/fig-ppsp.tex}

In addition to single PPSP queries, \ourlib{} also provides a simple and flexible interface to support \emp{batch PPSP queries}.
In many real-world scenarios,
multiple PPSP queries, either from independent users 
or as components in the same complex query (see examples below),
are required to be evaluated in a short time.
Batch PPSP has recently gained much attention and has been studied in many recent papers~\cite{knopp2007computing,geisberger2010engineering,geisberger2008contraction,xu2020simgq,xu2022simgq+,mazloumi2020bead,mazloumi2019multilyra,then2014more}.
Most of them are system work and study techniques such as caching or relabeling to take advantage of multiple SSSP searches.
However, we observe a gap in the study of batch PPSP queries:
this concept is used to refer to many different but closely related query types.
Here, we list several ``batch PPSP'' query types
with possible real-world applications as motivating examples below.

\begin{itemize}
  \item \emp{Single-source Many-target (SSMT)} and symmetrically Many-Source Single-Target: find the shortest paths from one source~$s$ to a set of targets $T\subseteq V$. Example: Finding the shortest path from the current location to any nearby Walmart.  
  \item \emp{Pairwise}: find shortest paths from all sources in $S\subseteq V$ to all targets in $T\subseteq V$. Example: Finding the shortest paths from all Walmart and all their warehouses in a region.
  \item \emp{Multi-Stop}: find shortest paths from $s=t_0$ visiting a list of targets $t_1,t_2,\dots,t_k$ in turn, i.e., a batch of $\{(t_0,t_1), (t_1,t_2),\dots (t_{k-1},t_k)\}$.
  Example: planning a trip with multiple stops on the way. 
  \item \emp{Subset APSP} (All-Pairs Shortest Paths): find pairwise shortest paths among a subset of vertices $V'\subset V$. Many algorithms use subset APSP as a building block, for instance, the ideas of hopsets~\cite{elkin2019hopsets,cao2020efficient} and landmarks~\cite{qiao2012approximate,tretyakov2011fast} for computing approximate shortest paths.
  \item \emp{Arbitrary Batch}: find shortest paths for a list of $s$-$t$ pairs, which can be disjoint or overlapping. 
\end{itemize}

Most existing work focuses on one or a subset of the special cases and lacks a general interface to 
leverage the shared information from arbitrary query batches.
Among the work we know of, some focus on pairwise queries~\cite{knopp2007computing,geisberger2010engineering,geisberger2008contraction}, and some others focus on SSMT queries~\cite{xu2022simgq+}.

One may also notice that due to the specificity of each query, the best strategy for them also remains elusive.
For instance, the SSMT query with a small number of targets may still benefit from running \bds{} from all
vertices, but when the target set $T$ becomes larger, one SSSP query from the source may give the best performance.
A more interesting example is the multi-stop queries, where
the best SSSP-based strategy may be running SSSP from \emph{every other vertex} $t_{2i+1}$
(when the graph is undirected) and combining the results of all $(t_{2i+1}, t_{2i}),(t_{2i+1},t_{2i+2})$ pairs.


\emp{The second contribution of \ourlib{} is a simple abstraction of arbitrary batch PPSP queries,
along with several efficient implementations.}
\ourlib{} formalizes all PPSP queries as a \emp{\querygraph{}} $\gq=(\vq,\eq)$, where $\vq$ contains all vertices in the query batch,
and each edge $(q_i,q_j)\in \eq$ indicates a PPSP query between $q_i$ and $q_j$.
\ourlib{} then analyzes the \querygraph{} and attempts to reuse the shortest path information of all the queried vertices.
Using the \querygraph, we develop algorithms for batch PPSP queries based on both \bds{} and SSSP.
For the \bds{}-based algorithm, we extend our PPSP framework to support batch queries by carefully redefining
the three functions mentioned above.
Specifically, this requires carefully setting the pruning conditions based on all queries in the batch.
As mentioned, when a vertex in $\gq$ is involved in many $s$-$t$ pairs,
using an SSSP from this vertex can likely give better performance.
A \naive{} approach is to run SSSP from all sources $s_i$
and save some work in the case that there are many targets associated with it.
By abstracting the \querygraph{}, a more interesting observation is that
the best SSSP-based strategy is to find the \emp{vertex cover}
of $\gq$ and run SSSP from the cover.
For example, for the multi-stop query, the \querygraph{} forms a chain, and the vertex cover consists of every other vertex on it.
This SSSP-based approach is integrated into \ourlib{}, and for certain special queries, 
it can outperform \bds{}.

For all our algorithms, we highlight their \emp{simplicity, generality}, as well as \emp{efficiency}. 
By modeling related queries into the framework in \cref{alg:ppsp}, each approach can be concisely described by
the instantiation with the three user-defined functions. 
The simplicity also enables high performance with relatively low coding effort, 
as most optimizations for parallel SSSP can be directly used. 

We evaluated \ourlib{} on both single and batch PPSP queries using 14 graphs with different types.  
For single PPSP queries, we tested \ourlib{} using various distance percentile (1\%, 50\%, and 99\% closest) query pairs. 
and compared it against two baselines, GraphIt (CGO'20)~\cite{zhang2020optimizing} and MBQ (SPAA'24)~\cite{zhang2024multi}. 
Both of them only support parallel single PPSP and \astar{} using unidirectional search.
In \cref{sec:exp}, we show our algorithm with bidirectional search achieves significant speedup over the baselines.
On average, across all tests, our \bds{} is 2.9$\times$ faster than GraphIt, and 6.8$\times$ faster than MBQ (even excluding MBQ's timeout cases). 
Our bidirectional \astar{} is 4.4$\times$ and 6.2$\times$ faster than the \astar{} in GraphIt and MBQ, respectively.
These results demonstrate the effectiveness of \ourlib{}, as well as the proposed algorithmic insights in this paper. 
For batched queries, we provide in-depth experimental evaluation to compare the two solutions in \ourlib{} (\bids{}-based and SSSP-based) on various query graph patterns. 
Our code is publicly available~\cite{ppspcode}.
\ifconference{We present more results and analyses in the full version of this paper~\cite{dong2025parallelfull}.}
We believe this is the first systematic study of \bds{}, \astar{}, and \bdastar{} for PPSP, and batch PPSP queries in the parallel setting.

%% file: figures/fig-ppsp.tex
\begin{figure*}
  \centering
  \includegraphics[width=\textwidth]{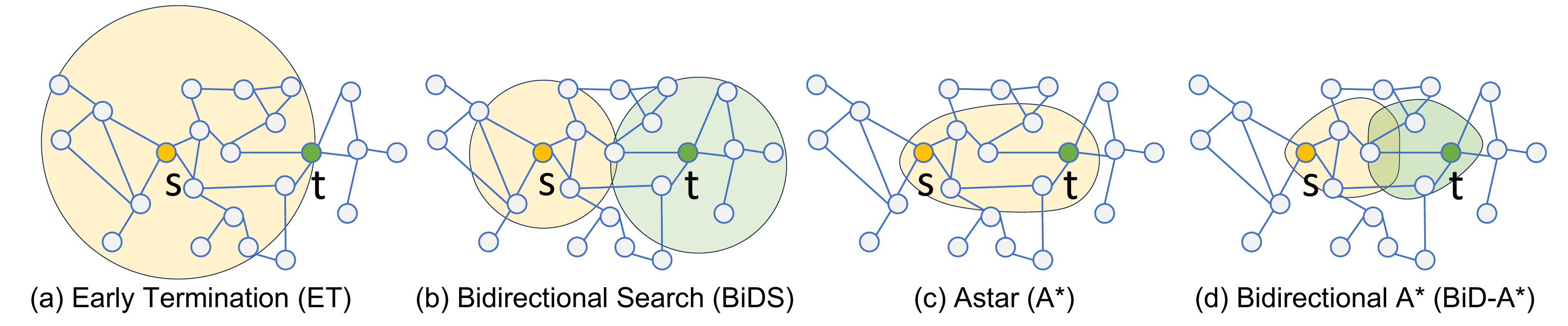}
  \caption{\textbf{Illustration for PPSP algorithms with early termination, bidirectional search, \astar, and bidirectional \astar{}}. 
  Details can be found in \cref{sec:algo}. (a) Early Termination (ET) runs SSSP from $s$ until $t$ is settled. All vertices with distances smaller than $t$ are processed. (b) Bidirectional Search (\bids{}) runs SSSP from both $s$ and $t$, until both sides settle a common vertex. 
  Each search only touches a small range of vertices around $s$ and $t$. (c) Astar (\astar{}) runs SSSP from $s$ with heuristic information. 
  The search is roughly ``guided'' towards $t$. (d) Bidirectional \astar{} (\bdastar) combines the advantages of both \bids{} and \astar{}. It runs SSSP from both $s$ and $t$
  and avoids searching vertices from $s$ that are away from $t$, and similarly for $t$. 
  }\label{fig:ppsp}
\end{figure*}

%% file: prelim.tex
\section{Preliminaries}\label{sec:prelim}

\myparagraph{Computational Models.}
We use the standard binary fork-join model~\cite{CLRS,blelloch2020optimal}.
We assume a set of \thread{}s that share a common memory.
A process can \forkins{} two child software \thread{s} to work in parallel.
When both children complete, the parent process continues.
We use the atomic operation \WriteMin{}$(p,\mathit{v})$, which reads the memory location pointed to by $p$ and writes value $v$ to it if $v$ is smaller than the current value.
It returns $\true{}$ if the update is successful and $\false{}$ otherwise.
\hide{The \defn{work} of an algorithm is the total number of instructions and
the \defn{span} (depth) is the length of the longest sequence of dependent instructions in the computation.}
\hide{We can execute the computation using a randomized work-stealing
scheduler in practice in $W/P+O(D)$ time \whp{} on $P$ processors~\cite{BL98,ABP01,gu2022analysis}.}

\input{tables/notations.tex}
\myparagraph{Notations.}
We consider a weighted graph $G=(V,E,w)$ with $n=|V|$, $m=|E|$, and an edge weight function $w:E\to \R^+$.
For $v \in V$, we define $N(v)=\{u\,|\,(v,u)\in E\}$ as the \defn{neighbor set} of $v$.
For PPSP queries, we use $s$ as the source and $t$ as the target. 
In our algorithms, we always maintain $\delta[v]$ as the \emp{tentative distance} from the source to vertex $v$,
which will be updated in the algorithm by the relaxations on $v$.
We also use $\delta[\vcopy{v}{\oplus}]$ to denote the distance for the same vertex $v$ from different sources.
For example, in bidirectional search, we use $\delta[v^{\langle+\rangle}]$ to denote the tentative distance from the source and $\delta[v^{\langle-\rangle}]$ from the target. 
We also denote $d(s,v)$ as the \defn{true distance} from $s$ to $v$.
With clear context, we simplify it to $d(v)$.
We say a vertex is \defn{settled} if its true shortest distance has been computed in $\delta[v]$. 
A list of notations is summarized in \cref{tab:notation}.

\input{pseudo-stepping}
\myparagraph{Stepping Algorithm Framework.}
The stepping algorithm framework proposed in~\cite{dong2021efficient},
is a high-level abstraction of many existing parallel SSSP algorithms,
including Dijkstra~\cite{dijkstra1959,brodal1998parallel}, Bellman-Ford~\cite{bellman1958routing,ford1956network}, $\Delta$-Stepping~\cite{meyer2003delta}, $\rho$-Stepping~\cite{dong2021efficient}, and more~\cite{blelloch2016parallel,Shi1999}.
The framework (see \cref{alg:stepping}) runs in steps and maintains a \defn{frontier} $F$ as the set of vertices to be processed.
In each step, it extracts vertices (\cref{line:step:extract}) that have tentative distances within a certain \textit{threshold} $\theta$ (\cref{line:step:threshold})  
from the frontier and relaxes their neighbors in parallel (\cref{line:step:innerloop2}).
If a vertex is successfully relaxed (\cref{line:step:writemin}),
it will be added to the next frontier (\cref{line:step:update}).
This framework applies to various algorithms by plugging in different functions for \extcond{}, 
which decides the threshold $\theta$ of the current step.
The algorithmic ideas in this paper can be combined with many SSSP algorithms in the stepping framework. 
In \ourlib{}, we use $\Delta^*$-stepping~\cite{meyer2003delta,dong2021efficient} as it is one of the most commonly used parallel SSSP algorithms.
In $\Delta^*$-stepping, the $i$-th call of $\extcond()$ returns $i\cdot \Delta$, where $\Delta$ is a user-defined parameter.

\myparagraph{Other Shortest-Path Algorithms.}
As one of the most studied problems, there have been many insightful algorithms designed for computing shortest paths.
We will review them in \cref{sec:related}.

\hide{
For example, Dijkstra~\cite{dijkstra1959} returns the minimum distance among the vertices in the frontier.
However, it limits parallelism because only few vertices have the minimum distance in the general case.
In contrast, $\rho$-stepping~\cite{dong2021efficient} returns the $\rho$-closest distances in the frontier to allow better parallelism,
where $\rho$ is a parameter decided by the hardware (i.e., the granularity to saturate all processors with sufficient tasks in each step).
Unlike other parallel SSSP algorithms, $\rho$-stepping is relatively insensitive to the parameter $\rho$ within a reasonably large range.
Therefore, we use $\rho$-stepping in the following sections to present our algorithm
to ease the effort to tune the parameters on different types of graphs with different weight ranges.
}

\hide{We note that being insensitive to the algorithm parameter is especially an advantage for PPSP algorithms,
as the behavior of the algorithm can be highly impacted by how far the queried vertices are from each other.
For other algorithms, such as $\Delta$-stepping, the best parameter $\Delta$ may vary significantly
for close-by or far-away vertex pairs.
} 

%% file: tables/notations.tex
\begin{table}[t]
  \small
  \centering
\hide{
\begin{tabular}{ll|clll}
  \toprule
  $G=(V,E)$ & : input graph & $\delta[v]$ & \multicolumn{3}{l}{: tentative distance of $v$} \\
  $N(v)$ & : neighbor set of $v$ & $d(u,v)$ & \multicolumn{3}{l}{: true distance from $u$ to $v$} \\
  $u$-$v$ & : path from $u$ to $v$ & $h(v)$ & \multicolumn{3}{l}{: heuristic of $v$} \\
  \end{tabular}%
}

\begin{tabular}{clclll}
  \toprule
  \multicolumn{6}{l}{\bf General Notations:} \\
  $G$ & : input graph $G=(V,E)$ & $\delta[v]$ & \multicolumn{3}{l}{: tentative distance of $v$} \\
  $N(v)$ & : neighbor set of $v$ & $d(u,v)$ & \multicolumn{3}{l}{: true distance from $u$ to $v$} \\
  $u$-$v$ & : path from $u$ to $v$ & $h(v)$ & \multicolumn{3}{l}{: heuristic of $v$} \\
  \end{tabular}%

\begin{tabular}{@{}rl|cl}
  \midrule
  \multicolumn{2}{l|}{\bf For single PPSP:} & \multicolumn{2}{l}{\bf For batch PPSP:} \\
  \multicolumn{1}{@{}c}{$s$} & : source & $\gq$ & : query graph $=\gq(\vq,\eq)$ \\
  \multicolumn{1}{@{}c}{$t$} & : target & $\nq(v)$ & : neighbor set of $v$ in $\gq$ \\
  \multicolumn{1}{@{}c}{$\vcopy{v}{+}$} & : $v$ from the source & $q_i$ & : the $i$-th source in $\vq$ \\
  \multicolumn{1}{@{}c}{$\vcopy{v}{-}$} & : $v$ from the target & $\vcopy{v}{i}$ & : $v$ from $q_i$ \\
  \multicolumn{1}{@{}c}{$\mu$} & : tentative distance & $\mu[\langle q_i,q_j \rangle]$ & : tentative distance from $q_i$ to $q_j$ \\
  & \ \ from $s$ to $t$ & $\mumax[i]$ & : search radius from source $q_i$ \\
  \bottomrule
\end{tabular}
\caption{\textbf{Notations}.}\label{tab:notation}

\hide{
\begin{tabular}{@{}c@{ }l@{ }|c@{ }l|c@{ }l}
  \toprule
  \multicolumn{2}{@{  }l|}{\bf General Notations:}&\multicolumn{2}{@{  }l|}{\bf For Single PPSP:} & \multicolumn{2}{@{  }l}{\bf For Batch PPSP:} \\
  $G$ & : input graph $G=(V,E)$  &\multicolumn{1}{@{ }c@{}}{$s$} & : source & $\gq$ & : query graph $\gq=(\vq,\eq)$\\
  $N(v)$ & : neighbor set of $v$ &  \multicolumn{1}{@{ }c@{}}{$t$} & : target & $\nq(v)$ & : neighbor set of $v$ in $\gq$ \\
  $u$-$v$ & : path from $u$ to $v$  &\multicolumn{1}{@{ }c@{}}{$\vcopy{v}{+}$} & : $v$ from the source & $q_i$ & : the $i$-th source in $\vq$ \\
  $\delta[v]$ & {: tentative distance of $v$} & \multicolumn{1}{@{ }c@{}}{$\vcopy{v}{-}$} & : $v$ from the target & $\vcopy{v}{i}$ & : $v$ from $q_i$ \\
   $d(u,v)$ & {: true distance from $u$ to $v$}& \multicolumn{1}{@{ }c@{}}{$\mu$} & : tentative distance & $\mu[\langle q_i,q_j \rangle]$ & : tentative distance from $q_i$ to $q_j$ \\
  $h(v)$ & {: heuristic of $v$}&\multicolumn{1}{@{ }c@{}}{~} & \ \ from $s$ to $t$ & $\mumax[i]$ &: search radius from source $q_i$ \\
  \bottomrule
\end{tabular}
  \caption{\textbf{Notations}.}\label{tab:notation}
}
\end{table}

%% file: pseudo-stepping.tex
\begin{algorithm}[t]
  \small
  \caption{The Stepping Algorithms Framework\label{alg:stepping}}
  \SetKwFor{parForEach}{ParallelForEach}{do}{endfor}
  \KwIn{A graph $G=(V,E,w)$ and a source node $s$.}
  \KwOut{The shortest path distances $\delta[\cdot]$ from $s$.}
  \DontPrintSemicolon
  $\delta[\cdot]\leftarrow +\infty$\\
  $\delta[s]\leftarrow 0$, $F.\lazyinsert{}(s)$\tcp*[f]{Initialize and add $s$ to the frontier $F$}\\
  \While{$|F|>0$\label{line:step:while-loop}} {
    $\theta \gets \extcond(F)$ \label{line:step:threshold}\tcp*[f]{Get the distance threshold based on the frontier}\\
    \tcp{process all $u$ with tentative distances no larger than $\theta$}
    \parForEach{$u\in F.\extract(\theta)$\label{line:step:extract}} {
        \parForEach{$v\in N(u)$\label{line:step:innerloop2}} {
        \If {$\WriteMin(\delta[v],\delta[u]+w(u,v))$ \label{line:step:writemin}} {
          $F.\lazyinsert{}(v)$\label{line:step:update} \tcp*[f]{Add $v$ to the frontier $F$ if $v\notin F$ previously}
        }
      }
    }
  }
  \Return {$\delta[\cdot]$}\\[.15em]
\end{algorithm} 

%% file: algo.tex
\section{Algorithms for Point-to-Point Shortest Paths}\label{sec:algo}

In this section, we present our point-to-point shortest path (PPSP) framework
and implementations of early termination (ET), \astar{}, \bds{}, and \bdastar{} based on this framework.
All these techniques have been shown to be effective for sequential PPSP.
Although early termination and \astar{} have also been studied in parallel settings,
we are aware of little work on using their bidirectional counterparts to achieve high-performance PPSP.
A key challenge we observed is that the efficiency of sequential \bds{} relies on strong ordering in vertex visits (Dijkstra’s order),
which is not preserved in parallel SSSP algorithms.
Parallel SSSP algorithms usually update vertices in batches or asynchronously, where strong ordering is not enforced.
Hence, in parallel, non-trivial adaptations are required to apply the ideas of \bds{}, \astar{}, and \bdastar{} effectively to improve performance.


To overcome these challenges,
we first design a PPSP framework adapted from the stepping algorithms framework in~\cite{dong2021efficient}.
Our framework eases the coding effort by abstracting the pruning conditions
and, more importantly, facilitates our algorithm on batch PPSP queries.
To incorporate bidirectional search, we design a pruning criterion independent of Dijkstra's order and present its adaptation for \bdastar{}.
Both our \bids{} and \bdastar{} fit seamlessly into our framework and significantly improve the performance.

Note that both bidirectional search and \astar{} are considered heuristics, which do not improve the worst-case bounds.
Our goal in this paper is to show that all these techniques can be the same effective in the parallel setting as in the classic sequential case.
The work and span bounds for \ourlib{} are upper-bounded by the underlying parallel SSSP algorithm~\cite{dong2021efficient}.

In the rest of this section, we first present our framework (\cref{sec:algo:framework}),
then discuss early termination (\cref{sec:algo:et}) and \astar{} (\cref{sec:algo:astar}) as examples to show the simplicity of our framework.
We then present our new algorithms on \bids{} (\cref{sec:algo:bds}) and \bdastar{} (\cref{sec:algo:biastar}) and also prove the correctness.

\subsection{Our Framework}\label{sec:algo:framework}
We present our framework in \cref{alg:ppsp}.
Under this framework, variants of our PPSP algorithms 
differ only in how they 
1) find the first frontier of a stepping algorithm, specified by the function $\initial$;
2) determine whether further search on a vertex can be pruned, specified by the function $\prune$; and
3) update the $s$-$t$ path distance accordingly, specified by the function $\update$.

\input{pseudo-ppsp.tex}
The algorithm begins by initializing the frontier and setting the source distance to zero using the $\initial$ function (\cref{line:ppsp:init})
and then proceeds in steps, similar to a regular stepping algorithm.
The main difference is that since searches can be from multiple sources (e.g., a bidirectional search or batch queries),
each vertex may appear in the frontier multiple times (each from a different source).
We use $\vcopy{u}{\oplus}$ to denote an element in the frontier,
where $u\in V$ is a vertex and $\oplus$ denotes the source of the search. 
For example, in bidirectional searches, a vertex $u$ can appear in both forward and backward searches.
We use $\vcopy{u}{+}$ and $\vcopy{u}{-}$ to denote these two cases, respectively.
As a result, the extracted vertices are those with the closest distances across multiple sources. 
If a vertex $\vcopy{u}{\oplus}$ is extracted, we first determine whether it should be pruned using the function $\prune$ (\cref{line:ppsp:check}).
If not, we process its neighbors regularly.
Note that whenever $\prune(v)$ is called, $v$ satisfies $\delta[v]\neq\infty$ since $v$ is either the source or has been updated by other vertices.
During the algorithm,
we maintain a global variable~$\mu$, which is the current shortest distance from $s$ to $t$.
If a vertex $v\in N(u)$ can be relaxed, we use $\update$ to update the current shortest distance $\mu$
with the path $s$-$v$-$t$, 
ensuring that $\mu$ eventually converges to the exact shortest distance from $s$ to $t$.
If the neighbor $v$ does not satisfy the pruning condition, it is added to the frontier (\cref{line:ppsp:add}).

One may note that updating $\mu$ requires priority updates, but the contention here is light---existing analysis~\cite{shun2013reducing} shows that the content can generally be considered as logarithmic.

In the following, we will present our four PPSP algorithms 
by showing how they instantiate the functions $\initial, \prune$, and $\update$.
We summarize the conditions in  \cref{tab:condition}. 

\input{tables/condition.tex}

\subsection{Early Termination (ET)}\label{sec:algo:et}
For an algorithm in the stepping algorithm framework, the tentative distances $\delta[v]$ for any vertex~$v$ only monotonically decreases.
Existing work~\cite{xu2019pnp,zhang2020optimizing} takes advantage of this monotonicity to prune the exploration of vertices $v$ with $\delta[v] \geq \mu$ as they do not contribute to the final distance.
This technique is often referred to as \textit{early termination (ET)} and has been used in existing algorithms, including PnP~\cite{xu2019pnp}, GraphIt~\cite{zhang2020optimizing}, and MBQ~\cite{zhang2024multi}.

Integrating ET to \ourlib{} is straightforward. 
The \initial{} function simply initializes the distance for $s$ as zero and adds it to the frontier.
Since searches all have the same source $s$, we omit the $\oplus$ parameter in \cref{tab:condition}. 
Recall that the algorithm maintains $\mu$ as the current best distance from $s$ to $t$.
The $\prune$ function skips a vertex $v$ with $\delta[v] \geq \mu$.
Accordingly, the \update{} function updates $\mu$ when a relaxation results in a smaller distance for $t$.

\hide{
In Dijkstra's, early termination can be easily applied because it visits vertices in ascending order of distances,
which means the algorithm can terminate after $t$ is popped from the priority queue.
However, in the parallel setting, there is usually no strict order of propagation,
which makes early termination complicated.
Existing work has studied various pruning conditions based on what the underlying SSSP algorithms are.
For example, PnP~\cite{xu2019pnp} is modified from parallel Bellman-Ford.
It simply prune vertices $v$ with $\delta[v] \geq \delta[t]$
since it cannot infer what vertices have been settled from Bellman-Ford.
As a $\Delta$-stepping-based algorithm which provides a loose ordering on the vertices being settled,
GraphIt~\cite{zhang2020optimizing} terminates when it enters iteration $i$ and $i\Delta$ is no less than $\delta[t]$.
}

\subsection{Astar (\astar{})}\label{sec:algo:astar}
\astar{} is widely used in pathfinding, informational search, and natural language processing~\cite{bast2016route,zeng2009finding,hart1968formal,nilsson2009quest,klein2003parsing}.
When used in graph PPSP queries, 
\astar{} employs a heuristic function $h(\cdot)$,
where $h(v)$ provides an
estimate (usually a lower bound) of the distance from $v$ to the target $t$.
The estimate can be the geometric distance (e.g., Euclidean distance) from $v$ to $t$.
As exemplified by Dijkstra's algorithm,
when selecting a vertex to search,
it chooses the vertex $v$ with the smallest value of $\delta[v]+h(v)$, until $t$ is settled.
This process is also referred to as best-first search.

The heuristic function usually satisfies two properties:
(1) \defn{Admissibility}:  $\forall v\in V$, $h(v)\leq d(v,t)$.
(2) \defn{Consistency}: $\forall (u,v)\in E$, $h(u)\leq w(u,v)+h(v)$. 
Note that when $h(t)=0$ (true for almost all commonly used heuristics), consistency implies admissibility~\cite{pearl1984heuristics}. 
In this paper, we assume the heuristic function satisfies consistency.
If $h(\cdot)$ is consistent, running any shortest path algorithm 
on $G=(V,E,w)$ with $h(\cdot)$ applied is equivalent to running it on
an induced graph $G'=(V,E,w')$ with modified weights $w'(u,v)=w(u,v)-h(u)+h(v)$.
The induced graph has two useful properties.

\begin{fact}\label{fact:consistency}
  Given a graph $G=(V,E,w)$ and a consistent heuristic $h(\cdot)$, 
  its induced graph $G'=(V,E,w')$ with the modified weight $w'(u,v)=w(u,v)-h(u)+h(v)$ has two properties: 
  1) For all $(u,v)\in E$, $w'(u,v)$ is nonnegative; and
  2) Given a source $s$, $d'(s,t)=d(s,t)-h(s)+h(t)$, where $d(s,t)$ and $d'(s,t)$ are the distances between $s$ and $t$ on $G$ and $G'$, respectively. 
\end{fact}

Thus, the induced graph modifies the shortest distance between $s$ and $t$ by a constant shift of $h(t)-h(s)$ while preserving the relative distances.
This guarantees the feasibility of running Dijkstra (and other algorithms) on $G'$,
with the extra benefit that the search prioritizes expanding vertices closer to the target.
\cref{fact:consistency} also implies that we can replace Dijkstra's algorithm with parallel SSSP algorithms in our framework, which yield efficient parallel \astar{} algorithms.

\hide{
and  as one can use Euclidean distances
as the heuristic function to estimate the lower-bound distance
from a vertex $v$ to the target $t$.
}

When implementing these algorithms in \ourlib{}, the only difference between ET and \astar{} is in the \prune{} function:
a vertex $v$ can be pruned if it satisfies $\delta[v]+h(v) \geq \mu$.
This is because $\delta[v]+h(v)$ provides a lower bound on the shortest path through $v$.
If it is already no less than $\mu$, search from $v$ will not improve the solution. 
Although \astar{} is simple and can be easily integrated into parallel SSSP algorithms,
a previous study~\cite{zhang2020optimizing} reports that \astar{} performs slower than ET.
We improve \astar{} with memoization, which enables it to outperform ET, as discussed in~\cref{sec:imp}.

\subsection{Bidirectional Search (\bds)}\label{sec:algo:bds}
\bds{} for PPSP has been widely studied sequentially~\cite{pohl1969bi,goldberg2006reach,luby1989bidirectional,goldberg2005computing,holte2016bidirectional}
and is shown to have better bounds for certain graph types and better practical performance than ET in general. 
As shown in \cref{fig:ppsp} (b), although \bds{} performs two searches from both $s$ and $t$,
the total search space can be smaller than a standard PPSP 
since both searches cover smaller radii.
The idea of most sequential \bds{} algorithms is to run Dijkstra's algorithm from both sides
due to the following theorem.

\begin{theorem}[\cite{pohl1969bi}]\label{thm:bi-dijkstra}
  Given a nonnegative weighted Graph $G=(V,E,w)$, a source $s$, and a target $t$,
  one can use Dijkstra's algorithm to search from both $s$ and $t$,
  alternate between the two searches in any way, and terminate when a vertex $v$ has been settled by both searches.
\end{theorem}

The algorithm is correct regardless of the alternation strategy used, 
and existing sequential studies mainly focus on finding better alternation strategies.
For example, Nicholson's algorithm~\cite{nicholson1966finding} selects the vertex with minimum tentative distance among both the forward and backward searches,
while Dantzig's algorithm~\cite{dantzig1963linear,goldberg2005computing} alternates strictly between the forward and backward searches.

Unfortunately, parallel SSSP algorithms achieve parallelism by processing multiple vertices simultaneously
and have limited information about which vertices have been settled.
Therefore, it has been open on how to efficiently use the bidirectional search to accelerate parallel PPSP.
To the best of our knowledge, the only parallel PPSP algorithm that applies \bds{} is from PnP~\cite{xu2019pnp}.
However, unlike sequential \bds{}, PnP only uses \bds{} in preprocessing to predict which direction (from $s$ or $t$) results in less computation
and then uses the standard unidirectional search with ET.


\myparagraph{Our \bds{}.}\label{sec:bids}
We design our \bds{} by plugging in the new $\prune$ and $\update{}$ functions in \ourlib{}, 
which are independent of Dijkstra's order.
Note that the key difference here is that, in parallel SSSP, an algorithm only knows the tentative distance on vertex $v$ but cannot determine whether and when it converges to the true distance. 
In contrast, in Dijkstra's order, every vertex popped from the priority queue is settled with its true distance.
Our \bds{} is based on an observation that, in \bids{}, given any upper bound $\mu$ of $d(s,t)$, 
we can skip a vertex $v$ with tentative distance $\delta[v]\geq \mu/2$.
This indicates that any further searches beyond this point are wasteful and thus pruned.
Unlike existing parallel PPSP work that does not apply or can only partially apply \bds{},
our algorithm can apply \bds{} throughout the entire search.
We prove the correctness of our algorithm in ~\cref{thm:bids-correctness}.

We now show how to implement \bds{} in our framework in~\cref{alg:ppsp}.
Similar to the sequential \bds{}, our \bds{} also starts a forward search from $s$ and a backward search from $t$,
so a vertex can have two copies in the frontier.
To distinguish these two copies, we use $\oplus \in \{+, -\}$ to denote the source:
$\vcopy{v}{+}$ represents the copy originating from $s$, and $\vcopy{v}{-}$ represents the copy from $t$.
We also use $\delta[\vcopy{v}{+}]=\delta_{s}[v]$ and $\delta[\vcopy{v}{-}]=\delta_{t}[v]$.
Since the searches from both $s$ and $t$ start with themselves, we first add $\vcopy{s}{+}$
and $\vcopy{t}{-}$ to the frontier and initialize their distances to zero.
Before our algorithm processes a vertex $\vcopy{v}{\cdot}$ from either direction, 
the $\prune$ function checks whether $\delta[\vcopy{v}{\cdot}]\ge \mu/2$ and skips it if so. 
Finally, when processing a vertex $v$, the $\update$ function always attempts to examine the path $s$--$t$--$v$, and
updates $\mu$ to $\delta[\vcopy{v}{+}]+\delta[\vcopy{v}{-}]$ if it provides a lower value.
Since the pruning condition depends only on the tentative distances,
it is oblivious to whether a vertex is settled.
We formally analyze the correctness of our approach with the theorem below.

\begin{theorem}\label{thm:bids-correctness}
  Given a nonnegative weighted undirected graph $G=(V,E,w)$, our \bds{} algorithm can correctly compute the shortest $s$-$t$ distance 
  by running any stepping algorithm on parallel SSSP 
  from both $s$ and $t$ and skipping (not relaxing the neighbors of) any vertex $v$ with a tentative distance $\delta[v]\geq \mu/2$, where $\mu$ is the current shortest distance between $s$ and $t$, updated in \cref{alg:ppsp}.
\end{theorem}

\begin{proof}\label{proof:bids}
  Recall that $\delta[v]$ is the tentative distance and finally converges to the true distance $d(v)$, so we have $\delta[v]\geq d(v)$.
  Assume the true distance between $s$ and $t$ is $\mu'=d(s,t)$.
  Note that when $\mu$ has reached $\mu'$ (i.e., $\mu=\mu'$) in the algorithm, then the shortest distance $d(s,t)$ has been computed.
  
  Hence, we now focus on the case when $\mu'<\mu$. 
  As presented in~\cref{fig:bids}, assume one of the shortest paths 
  between $s$ and $t$ is $P=\{s, v_1, \cdots, v_{m-1}, v_m, v_{m+1}, \cdots, t\}$.
  Clearly, for all $v\in P$, $d(s,v)+d(v,t)=\mu'$.
  We assume vertex $v_m$ divides the path into two halves, $P_f=\{s, v_1, \cdots, v_{m-1}\}$ and $P_b=\{v_{m+1}, \cdots, t\}$,
  such that vertex $v\in P_f$ satisfies $d(s,v)\le \mu'/2<\mu/2$.
  Similarly, vertex $v\in P_b$ satisfies $d(s,v)\ge\mu'/2$,
  which implies $d(v,t)=\mu'-d(s,v)\le \mu'/2<\mu/2$.
  
  We prove by induction that all true distances on the path will be computed, and the true distance $d(s,t)$ will be correctly updated to the variable $\mu$ to be $\mu'$ and returned.
  In the forward search, $s$ is the source with $\delta[\vcopy{s}{+}]=d(s,s)=0$, which is the base case and not pruned. 
  Hence, $s$ updates $v_1$ in the first round with $\delta[\vcopy{v_1}{+}]=d(s, v_1)<\mu/2$ in the forward search.
  Similarly, all $v_i\in P_f\backslash \{s\}$ are updated by $v_{i-1}$ with $\delta[\vcopy{v_i}{+}]\le \mu'/2\le \mu/2$ and thus not pruned.
  Finally, $v_{m-1}$ updates $\delta[\vcopy{v_m}{+}]$ to the true distance $d(s,v_m)$.
  Similarly, none of the vertices on $P_b$ are pruned and will finally get their true distance. 
  Finally, $v_{m+1}$ updates $\delta[\vcopy{v_m}{-}]$ to the true distance $d(v_m,t)$.
  Since we have the true distance of $v_m$ in both forward and backward searches, we can obtain the true distance between $s$ and $t$ by concatenating the path $s$--$v_m$--$t$.
  The return value $\mu$ is correctly updated by the $\update{}$ function with $\mu=\delta[\vcopy{v}{+}]+\delta[\vcopy{v}{-}]=\mu'$, when the later direction is successfully updated.
\end{proof}

Despite its conceptual simplicity, this approach demonstrates very good performance in our experiments 
and readily fits into our framework. 
More importantly, we will show that this strategy facilitates our \bdastar{} in \cref{sec:algo:biastar} and our batch PPSP algorithms in \cref{sec:framework}. 
Surprisingly, we did not find existing work that employs such a simple heuristic for \bids{}.
We attribute such a gap to the inherent difference between sequential and parallel SSSP---sequentially,
since only \emph{one} vertex is processed at a time, it is natural (and easy) to find a settled vertex to explore, and thus no ``pruning'' is needed. 
We believe the simplicity is particularly interesting as an example of \emph{parallel thinking}---effective parallelization does not always add significant complication to a sequential algorithm, and can be made simple by a different way of thinking. 

\input{figures/fig-bids.tex}

\subsection{Bidirectional Astar (\bdastar{})}\label{sec:algo:biastar}
With \bds{} and \astar{}, it is natural to consider combining them as the bidirectional $A^*$ (\bdastar{}). 
However, \bdastar{} is challenging even in the sequential setting. 
Given two valid heuristic functions $h_s(v)$ estimating the distance to the source 
and $h_t(v)$ estimating the distance to the target,
let $h_F(\cdot)$ be the heuristic used in the forward search (from $s$)
and $h_B(\cdot)$ the heuristic in the backward search (from $t$). 
A simple approach is to use $h_F(\cdot)=h_t(\cdot)$ for the forward search and $h_B(\cdot)=h_s(\cdot)$ for the backward search and terminate when the two searches touch.
Unfortunately, this solution is incorrect~\cite{goldberg2005computing,ikeda1994fast}.
Intuitively, as in \astar, we can consider the induced graph $G'=(V,E,w')$ with modified edge weights $w'(u,v)=w(u,v)-h(u)+h(v)$.
In the forward search, an edge $(u,v)$ has weight $w'(u,v)=w(u,v)-h_F(u)+h_F(v)$,
while its mirrored edge in the backward search has weight $w'(v,u)=w(v,u)-h_B(v)+h_B(u)$,
which are not equal.
This difference breaks the consistency of the heuristic
and thus leads to incorrect results in this simple solution. 
To fix this issue, the most commonly used technique~\cite{goldberg2005computing} is to 
ensure $h_B(v)+h_F(v)=0$ for any vertex $v\in V$. 
Therefore, the modified edge weights $w'$ are identical in both forward and backward searches. 
We will discuss how our solution achieves a consistent heuristic for \bids{} below. 

As with \bids{}, the ideas in all sequential \bdastar{} rely on using Dijkstra's order. 
To parallelize this idea, we also combine \bdastar{} with the high-level idea mentioned in \cref{sec:algo:bds}
and prove that it can be combined with any consistent heuristic functions.

\myparagraph{Our Parallel Bidirectional \astar{}.}
To enable \bdastar{}, our high-level idea is to select the appropriate \emph{consistent} heuristics for both forward and backward searches, 
ensuring that our algorithms, 
when running on the original graph with the heuristics,
behave as if it were running on a graph (i.e., the induced graph) without the heuristics.
This allows us to directly apply our \bids{} pruning criterion on the induced graph, as it effectively appears heuristic-free.

To obtain consistency, we adapt the heuristics from previous work~\cite{goldberg2005computing,ikeda1994fast}. 
We set $h_{F}(v)=(h_t(v)-h_s(v))/2$ for the forward search 
and $h_{B}(v)=(h_s(v)-h_t(v))/2$ for the backward search,
which implies $h_{F}(v)+h_{B}(v)=0$ and thereby preserves heuristic consistency.
The intuition here is that, 
rather than directing the forward search towards the target and the backward search towards the source, 
both searches are guided towards the ``perpendicular bisector'' region between the source and target,
which is essentially a set of points equidistant from the source and target.

We now present how to implement our \bids{} under the framework in~\cref{alg:ppsp}.
We first add $\vcopy{s}{+}$ and $\vcopy{t}{-}$ to the frontier and initialize their distances to zero.
These steps are the same as in our \bids{}.
Note that our algorithm runs in steps,
and in each step, it selects a set of vertices within a certain distance threshold (\cref{line:ppsp:extract} in \cref{alg:ppsp}).
In \bdastar, we use the sum of the tentative distance and the heuristic function instead of only the tentative distance.
In addition, when our algorithm is about to scan a vertex $\vcopy{v}{+}$ in the forward search,
the $\prune$ function checks if $\delta[\vcopy{v}{+}]+h_F(v)\ge \mu/2$ and skips it if the condition holds.
Similarly, for a vertex $\vcopy{v}{-}$ in the backward search,
it checks whether $\delta[\vcopy{v}{-}]+h_B(v)\ge \mu/2$ instead.
Finally, when a vertex $v$ is reached from both directions, an $s$-$t$ path is found by concatenating $s$--$v$ and $v$--$t$,
so the $\update$ function updates the current $s$-$t$ distance $\mu$ with $\delta[\vcopy{v}{+}]+\delta[\vcopy{v}{-}]$,
which will eventually converge to the true distance.

\hide{
Our \bdastar{} algorithm also maintains the current best distance $\mu$, 
and updates it whenever both forward and backward searches visit the same vertex.
This part is consistent with our \bids{} algorithm.
As discussed above, we need to ensure the heuristics for the forward and backward search satisfy $h_F(v)+h_B(v)=0$ for any $v$. 
We adapt the consistent heuristic functions from previous work~\cite{goldberg2005computing,ikeda1994fast}.
We set $h_{F}(v)=(h_t(v)-h_s(v))/2$ for the forward search 
and $h_{B}(v)=(h_s(v)-h_t(v))/2$ for the backward search,
so that $h_{F}(v)+h_{B}(v)=0$ and the consistency of the heuristic is preserved.
The high-level idea behind the new heuristic functions is, 
instead of prioritizing the forward search expanding to the target and the backward search to the source,
we prioritize both searches expanding towards the ``perpendicular bisector'' region of the source and the target,
which can also be considered as the midpoints of the source and target.
We now show that combining this heuristic and our new pruning strategy not only 
effectively guides the searches to prioritize the midpoints, 
it also provides an extra stopping condition: \bdastar{} can prune forward (or backward) searches on a vertex $v$ if $\delta[\vcopy{v}{+}]+ h_{F}(v)\geq \mu/2$ (or $\delta[\vcopy{v}{-}]+ h_{B}(v)\geq \mu/2$).
}

We now show that our algorithm can correctly compute the shortest paths in \cref{thm:bdastar-correctness} with \emph{any consistent} heuristic functions.
\ifconference{For the page limit, we present the proof in the full paper~\cite{dong2025parallelfull}.} 
\iffullversion{The proof is presented in the Appendix. }
\hide{We will show that running \bdastar{} on the input graph $G$ is equivalent to running \bids{} on an auxiliary graph $G'$, and our algorithm can correctly compute the shortest paths on $G'$, and thus $G$.}

\begin{theorem}\label{thm:bdastar-correctness}
  Given a nonnegative weighted undirected graph $G=(V,E,w)$, a source $s$, a target $t$, and a consistent heuristic function $h(\cdot)$ (i.e, for all $(u, v)\in E$, $h(u)\leq w(u,v)+h(v)$),
  our \bdastar{} algorithm can correctly compute the shortest $s$-$t$ distance 
  by running any stepping parallel SSSP algorithms from both $s$ and $t$,
  and skipping (i.e., not relaxing the neighbors of) 
  any vertex $v$ having a tentative distance $\delta[\vcopy{v}{+}]+h(v)\geq \mu/2$ if $v$ originates from the source, or $\delta[\vcopy{v}{-}]-h(v)\geq \mu/2$ if $v$ originates from the target,
  where $\mu$ is the current shortest $s$-$t$ distance that is updated in~\cref{alg:ppsp}.
\end{theorem}

\hide{
\begin{proof}\label{proof:bdastar}
  Assume one of the shortest paths between $s$ and $t$ is $v_0$--$v_1$--$\cdots$--$v_k$, where $v_0=s$ and $v_k=t$.
  We also assume $d(\cdot)$ is the final distances.
  Since $\delta[\cdot]$ is the tentative distances, $d(v) \leq \delta[v]$ for an $v$.
  For any $v_i$ $(0\leq i\leq k)$, $d(v_i^+)+d(v_i^-) = \mu$.
  Therefore, for any vertex $v_i$, 
  \xiaojun{Figure out the conditions.}
  \hide{
  1) if $d(v_i^+) < \mu/2$, $v_i$ will not be skipped from the subsequent search path from $s$. 
  2) if $d(v_i^-) < \mu/2$, $v_i$ will not be skipped from the subsequent search path from $t$. 
  3) if $d(v_i^+) =d(v_i^-) = \mu/2$, although neither $v_i^+$ or $v_i^-$ will not continue the 
  searches, these two vertices have been settled to their final distances. 
  }
  Therefore, the skipping condition provided by our algorithm does not change the exact path distance between the given $s$ and $t$.
\end{proof}
}

We note that the heuristic functions used in our paper are relatively straightforward.
More advanced \bdastar{} heuristics have been developed, 
such as BAE$^*$~\cite{sadhukhan2013bidirectional} and DIBBS~\cite{sewell2021dynamically}.
For a broader discussion on heuristic strategies in bidirectional search, we refer the reader to the survey~\cite{sturtevant2018brief}.

In~\cref{sec:exp}, we will  show that combining \bids{} and \astar{} into \bdastar{}
significantly improves the performance for PPSP.

%% file: pseudo-ppsp.tex
\begin{algorithm}[t]
  \small
  \caption{\small Point-to-Point Shortest Path Algorithm Framework} \label{alg:ppsp}
  \SetKwFor{parForEach}{ParallelForEach}{do}{endfor}
  \KwIn{A graph $G=(V,E,w)$.}
  \KwOut{
  Shortest path between a given pair of vertices or multiple such query pairs.
  }
  \DontPrintSemicolon
  $\mu\gets +\infty, \delta[\cdot]\gets +\infty$ \\
  $\initial()$\label{line:ppsp:init}\tcp*[f]{Initialize the frontier $F$ and the distance $\delta[\cdot]$ for the sources} \\
  \While{$|F|>0$\label{line:while-loop}} {
    $\theta \gets \extcond(F)$ \label{line:ppsp:extract} \\
    \tcp{$\oplus$ denotes the source of $u$}
    \parForEach{$\vcopy{u}{\oplus} \in F.\extract(\theta)$} { 
      \tcp{Prune the search at $\vcopy{u}{\oplus}$ if it cannot contribute to $\mu$}
      \If{$\neg \prune(\vcopy{u}{\oplus})$} {\label{line:ppsp:check}
        \parForEach{$v\in N(u)$}{
          \If {$\WriteMin(\delta[\vcopy{v}{\oplus}],\delta[\vcopy{u}{\oplus}]+w(u,v))$\label{line:ppsp:relax}} {
            \tcp{Update $\mu$ using the path through $v$}
            $\mu = \update(\vcopy{v}{\oplus})$\label{line:ppsp:update}\\
            \lIf{$\neg \prune(\vcopy{v}{\oplus})$\label{line:ppsp:prune}} {
              $F.\lazyinsert{}(\vcopy{v}{\oplus})$\label{line:ppsp:add}
            }
          }
        }
      }
    }
  }
  \Return $\mu$
\end{algorithm}

%% file: tables/condition.tex
\newcommand{\Gq}{G_q}
\newcommand{\Vq}{V_q}
\newcommand{\Eq}{E_q}
\newcommand{\qq}{q}

\begin{table*}[htbp]
  \small
  \centering

\begin{tabular}{llll}
      \toprule
            & \multicolumn{1}{l}{\boldmath{}\textbf{\textsc{Init}$()$}\unboldmath{}} & \boldmath{}\textbf{\textsc{Prune}$(v)$}\unboldmath{} & \boldmath{}\textbf{\textsc{UpdateDistance}$(v)$}\unboldmath{} \\
      \midrule
      \midrule
      \multicolumn{4}{@{}l}{\bf \boldmath Single-Directional Searches for PPSP: Early Termination (ET) and \astar \unboldmath } \\ 
      \multicolumn{4}{@{}l}{For source $s\in V$ and target $t\in V$, $\mu$ maintains the current shortest distance between $s$ and $t$. 
      }
      \\ 
      \midrule
      \multicolumn{1}{l}{\textbf{~~ ET}} & $\delta[s]\leftarrow 0$, $F.\lazyinsert{}(s)$ & $\delta[v] \geq \mu$ & \multicolumn{1}{l}{If $v=t$: $\writemin(\mu,\delta[v])$} \\
      \midrule
      \multicolumn{1}{l}{\textbf{~~ \astar}} & $\delta[s]\leftarrow 0$, $F.\lazyinsert{}(s)$ & $\delta[v]+h(v) \geq \mu$ & \multicolumn{1}{l}{If $v=t$: $\writemin(\mu,\delta[v])$} \\
      \midrule
      \midrule
      \multicolumn{4}{@{}l}{\bf \boldmath Bi-Directional Searches for PPSP: Bi-Directional Searches (BiDS) and Bi-Directional \astar{} (\bdastar) \unboldmath}\\ 
      \multicolumn{4}{@{}l}{Each vertex $v\in V$ has two copies: $\vcopy{v}{+}$ denotes the search from the source $s$, and $\vcopy{v}{-}$ denotes the search from the target $t$}\\
      \multicolumn{4}{@{}l}{For source $s\in V$ and target $t\in V$, $\mu$ maintains the current shortest distance between $s$ and $t$. 
      }
      \\ 
      \midrule
      \multirow{2}[2]{*}{\textbf{~~ \bds}} & \multicolumn{1}{p{11em}}{$\delta[\vcopy{s}{+}]\leftarrow 0$, $F.\lazyinsert{}(\vcopy{s}{+})$} & \multirow{2}[2]{*}{$\delta[\vcopy{v}{\cdot}] \geq \mu/2$} & \multirow{2}[2]{*}{$\writemin(\mu, \delta[\vcopy{v}{+}]+\delta[\vcopy{v}{-}])$} \\
            & \multicolumn{1}{p{11em}}{$\delta[\vcopy{t}{-}]\leftarrow 0$, $F.\lazyinsert{}(\vcopy{t}{-})$} &       &  \\
      \midrule
      \multirow{2}[2]{*}{\textbf{~~ \bdastar}} & \multicolumn{1}{p{11em}}{$\delta[\vcopy{s}{+}]\leftarrow 0$, $F.\lazyinsert{}(\vcopy{s}{+})$} & For \textsc{Prune}$(\vcopy{v}{+})$: $\delta[\vcopy{v}{+}]+h_{F}(v) \geq \mu/2$ & \multirow{2}[2]{*}{$\writemin(\mu, \delta[\vcopy{v}{+}]+\delta[\vcopy{v}{-}])$} \\
            & \multicolumn{1}{p{11em}}{$\delta[\vcopy{t}{-}]\leftarrow 0$, $F.\lazyinsert{}(\vcopy{t}{-})$} & For \textsc{Prune}$(\vcopy{v}{-})$: $\delta[\vcopy{v}{-}]+h_{B}(v) \geq \mu/2$ &  \\
      \midrule
      \midrule
      \multicolumn{4}{@{}l}{\bf \boldmath Multi-Directional Searches for \emph{Batch} PPSP (Multi-PPSP)\unboldmath}\\ 
      \multicolumn{4}{@{}l}{Let $\Gq=(\Vq,\Eq)$ be the \querygraph{}. Each vertex $v\in V$ has $|\Vq|$ copies: $\vcopy{v}{i}$ denotes the search to vertex $v$ from the $i$-th vertex $q_i\in \vq$.}\\
      \multicolumn{4}{@{}l}{For $\langle \qq_i,\qq_j \rangle\in \Eq$ (i.e., the distance between $\qq_i$ and $\qq_j$ is queried), 
      $\mu[\langle \qq_i,\qq_j \rangle]$ maintains the current shortest distance between $\qq_i$ and $\qq_j$. 
      }\\\midrule
      \multirow{4}[2]{*}{\textbf{\tbf{~~Multi-PPSP}}} &  & \multirow{4}[2]{*}{For \textsc{Prune}$(\vcopy{v}{i})$: $\delta[\vcopy{v}{i}]\geq \mumax[i]/2$} & For \textsc{Update}$(\vcopy{v}{i})$, let $q_i$ be the $i$-th vertex in $\Vq$: \\
      &\multirow{1}[2]{*}{For $\qq_i \in \Vq$: }&&\quad For each $\qq_j \in \nq(q_i)$:\\
      &\multirow{1}[2]{*}{\quad$\delta[\vcopy{\qq_i}{i}]\leftarrow 0$, $F.\lazyinsert{}(\vcopy{\qq_i}{i})$}&&\qquad$\writemin(\mu[\langle \qq_i,\qq_j \rangle], \delta[\vcopy{v}{i}]+\delta[\vcopy{v}{j}])$\\
            &       &       & \qquad $\writemin(\mumax[i], \max\{\mu[\langle \qq_i,\qq_j \rangle]~|~\qq_j\in \nq(\qq_i)\})$ \\
      \bottomrule
      \end{tabular}%

  \caption{\textbf{\textbf{\textsc{Init}}, \textbf{\textsc{Prune}}, and \textbf{\textsc{Update}} functions of our PPSP algorithms.} \normalfont
  The superscript $i$ of a vertex $v$ indicates $v$ is searched from $i$.
  ``ET'': early termination (\cref{sec:algo:et}).
  ``\astar'': astar search (\cref{sec:algo:astar}).
  ``\bds'': Bidirectional search (\cref{sec:algo:bds}).
  ``\bdastar'': bidirectional astar search (\cref{sec:algo:biastar}).
  $h(v)$ is the heuristic of $v$, e.g., based on the geometric distance from $v$ to $s$ and $t$. 
  ``$h_{F}(v)$'': heuristic used in the forward search.
  ``$h_{B}(v)$'': heuristic used in the backward search.
  Some examples of the heuristic functions are provided in~\cref{sec:algo}. 
  For all single PPSP algorithms, we maintain a global variable $\mu$,
  which is the current best distance from $s$ to $t$. 
  The $\update{}$ function will update this value, which finally converges to the true distance from $s$ to $t$.
  Similarly, in the Multi-PPSP algorithm, for each PPSP query $\langle q_i,q_j \rangle \in \Eq$ in the batch,  
  we use $\mu[\langle \qq_i, \qq_j\rangle]$ to maintain the current shortest distance between $\qq_i$ and $\qq_j$. 
  We also maintain an array of $\mumax[i]=\max\{\mu[\langle \qq_i,\qq_j \rangle]~|~\qq_j\in \nq(\qq_i)\}$, which is the farthest distance that needs to be searched from the $i$-th source. 
  }

  \label{tab:condition}

\end{table*}%

%% file: figures/fig-bids.tex
\begin{figure}
  \centering
  \includegraphics[width=\columnwidth]{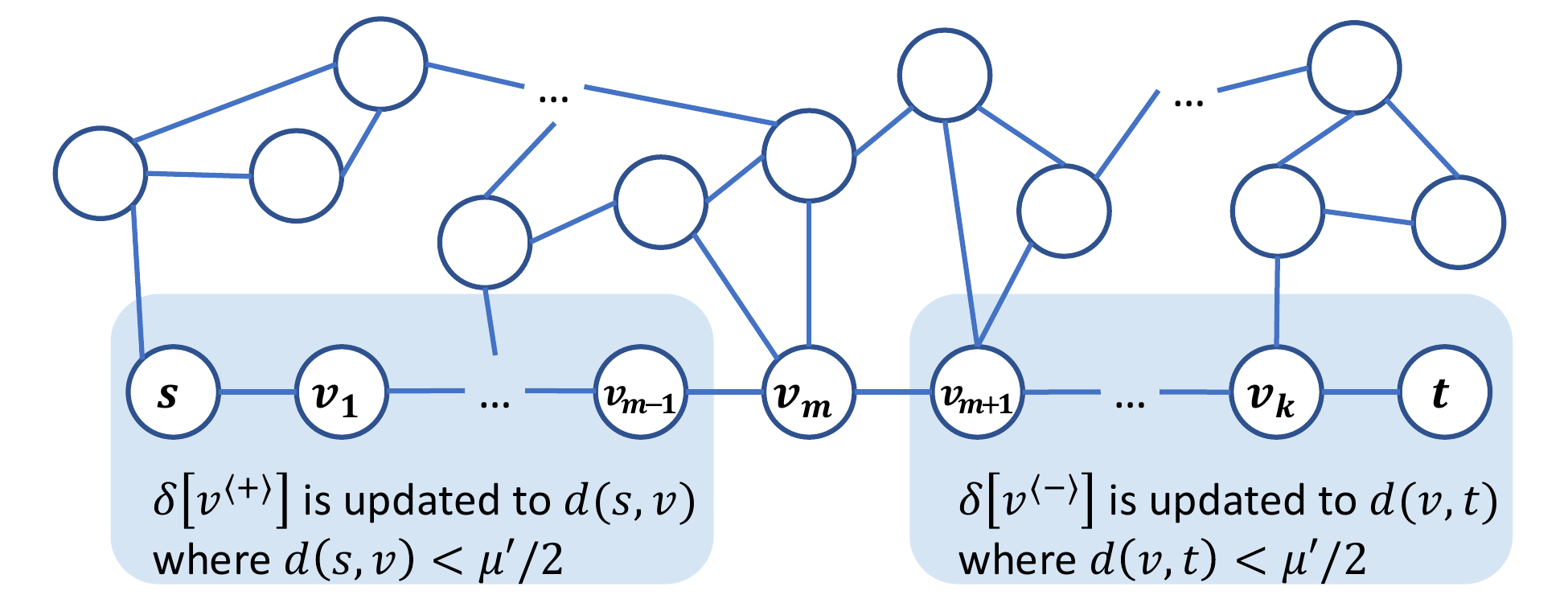}
  \caption{\textbf{Illustration for the proof of \cref{thm:bids-correctness} for our \bids{} algorithm.}
  }\label{fig:bids}
\end{figure}

%% file: framework.tex
\section{Batch PPSP Queries in Parallel}\label{sec:framework}

In practical applications, multiple PPSP queries, either independent or sharing common vertices,
can be required in a batch or a complex task. 
A list of possible query types was reviewed in \cref{sec:intro}.
Running multiple queries in a batch can be helpful in two aspects.
First, it can potentially reduce some redundant searches if the set of queries share common vertices.
Second, in the parallel setting, batching multiple queries brings up more work, which
can saturate processors to achieve better parallelism.

Existing work has considered batching multiple PPSP queries in sequential~\cite{knopp2007computing,geisberger2010engineering,geisberger2008contraction}, parallel~\cite{xu2020simgq,xu2022simgq+}, and distributed settings~\cite{mazloumi2020bead,mazloumi2019multilyra}.
Our work differs from them in two aspects.
First, each of them usually only focuses on certain query types, such as pairwise queries~\cite{knopp2007computing,geisberger2010engineering,geisberger2008contraction} or single-source multiple-target (SSMT) queries~\cite{xu2022simgq+}.
Second, most of them are system-level improvements for caching or communication cost by techniques such as relabeling, cache line alignment, and some smart vertex marking;
the backbone is still some existing PPSP algorithms.
In this paper, we focus on algorithmic improvements such as \bds{}.

\hide{
Unlike existing PPSP algorithms that only support single query or some certain types of multiple queries,
\ourlib{} is more general in that it supports \emph{arbitrary queries},
which is achieved by \emph{decomposing} and \emph{aggregating}.
Surprisingly, with these two extra steps, \ourlib{} can naturally fit into our PPSP framework.
\xiaojun{Mention our algorithm is designed for symmetric graphs.}
}

In \ourlib{}, we aim to support \emph{arbitrary batch queries}.
We take a sequence of source-target pairs and compute the PPSP among them.
A key observation here is that the vertices in the queried batch are often highly relevant:
as mentioned in \cref{sec:intro}, these queries may come from
the same complex query, and the source and target can often overlap.
For simple query types such as SSMT, it is easy to reuse the computation since they are all from the same source,
but for more complicated cases, it is non-trivial to do so effectively.
Our goal is to 1) reuse the computation among the source and/or target vertices as much as possible,
2) keep the implementation reasonably simple, and 3) integrate the useful techniques such as \bds{} proposed in \cref{sec:algo}.
In the following, we first discuss how to preprocess the input queries in \cref{sec:multi:preprocess} and present our \bds{}-based algorithms in \cref{sec:multi:bids}.
We also show a baseline algorithm in \cref{sec:multi:sssp}.

\hide{
To do this, we abstract all the queries in a \emp{\querygraph{}} $\gq=(\vq, \eq)$.
Using the structural information of $\gq$, we can understand how the queries overlap with each other and
reuse the search information.
Below, we elaborate on how to construct the \querygraph and how to use it to achieve different approaches based on both SSSP and \bds{}.
}

\hide{
As mentioned in \cref{xxx}, our algorithm focus on symmetric graphs because
1) most real-world road networks are symmetric;
2) our algorithm are more helpful when the graphs are symmetric, since the forward and backward
searches are the same for a certain vertex.}

\subsection{Constructing the \titlecap{\querygraph{}}}\label{sec:multi:preprocess}
One of the advantages of batch PPSP queries is that they can potentially reduce redundant searches.
Therefore, the first step of \ourlib{} is to decompose the queries to identify common vertices in the given queries.
Given a set of queries $(s, t)\in Q$, we first generate a \querygraph{} $G_q=(V_q, E_q)$
such that each query corresponds to an edge in the \querygraph{}.
As presented in~\cref{fig:multi-ppsp}, in the \querygraph{}, each vertex $q\in V_q$ denotes a source or target in the query batch,
and each edge $(s, t)\in E_q$ denotes a single $s$-$t$ query in the batch.
After constructing the \querygraph, a batch of general queries boils down to a group of single-source multiple-target (SSMT) queries
starting from each $q\in V_q$ and their targets are the neighbors in the \querygraph{} (denoted as $N_q(q)$).
The structure \querygraph{} will help \ourlib{} understand the way to reuse the searches in the batch.

For the special queries mentioned in \cref{sec:intro}, each of them has a certain pattern of the \querygraph{}.
For example, a batch of SSMT queries (or symmetrically many-source single-target) forms a star \querygraph{}.
The pairwise query batch forms a complete bipartite graph as the \querygraph{}.
A simple multi-stop query can be represented as a \querygraph{} of a chain.
When each stop comes with multiple options, there may also be forks branching from a chain.
A subset APSP forms a clique \querygraph{} on a subset of vertices.
More generally, an arbitrary batch is just a general graph on all vertices (both sources and targets) involved in the batch.

\input{figures/fig-multi-ppsp.tex}
\subsection{\bds-Based Solutions}\label{sec:multi:bids}
As discussed in \cref{sec:algo}, \bds{} can greatly improve performance in a single PPSP query, so algorithmically, it is interesting to see how to integrate \bds{} in the batch setting.
A straightforward approach is to apply our parallel \bids{} algorithm to each query one-by-one serially,
which we refer to as the \emp{Plain} approach.
To further increase parallelism, we can instead run our parallel \bids{} algorithm on all queries concurrently, which we refer to as the \emp{Plain$^*$} approach.
However, a batch of queries often share common vertices.
After revealing the common vertices, we can reduce the redundant searches originating from the same sources.
More importantly, in \bds{}, a target can also be considered as a source on undirected graphs. 
When each vertex can be either source or target and can be involved in multiple PPSP queries, we need to carefully
redefine the conditions for pruning and updating the answers.
We refer to this approach as the \emp{Multi-\bds{}} (or \emp{multi} for short) approach.
Our solution in \ourlib{} is as follows.
Since a vertex can be reached from multiple sources,
for each vertex $u$ from the original graph,
we maintain a copy of it from each of its neighbors in the \querygraph{} $\gq$.
We use $\vcopy{u}{i}$ to denote that vertex $u$ is reached from $q_i\in \vq$,
and $\delta[\vcopy{u}{i}]$ stores the tentative distance from $q_i$ to $u$ (i.e., $\delta[\vcopy{u}{i}]=\delta_{q_i}[u]$).
Our pruning idea is based on pruning the search from a vertex $q_i\in \vq$ when
it passes the farthest midpoint of $q_i$ and $v\in N_q(q_i)$. 

Although slightly more complicated than a single PPSP query,
\ourlib{} also follows the PPSP framework in \cref{alg:ppsp}.
To effectively make use of the \bds{}, we search from all vertices in the \querygraph{}.
The \initial{} function starts with adding $q_i\in V$ to the frontier,
as all searches start from themselves.
Therefore, they add their own copy, which is $\vcopy{q_i}{i}$
with distance zero, to the frontier.

The problem boils down to setting the \prune{} and \update{} functions correctly,
so we can avoid redundant searches as much as possible.
Specifically, \ourlib{} maintains the tentative distances $\mu[\langle q_i,q_j \rangle]$ for each edge $(q_i, q_j)\in \eq$,
which are also the outputs of our algorithm.
In addition to the best answer for each edge (i.e., query),
our algorithm also maintains a ``search radius'' $\mumax[i]$ for each $q_i\in \vq$,
which indicates that the search starting from the source $q_i$ can be pruned at a vertex $u$ if $\delta[\vcopy{u}{i}] \geq \mumax[i]/2$,
where $\mumax[i]$ can be computed by $\mumax[i]=\max\{\mu[\langle q_i,q_j \rangle]~|~q_j\in \nq(q_i)\}$.
If the pruning condition holds, further search on vertex $v$ will not contribute to the final answer to any queries,
so we can safely skip $\vcopy{v}{i}$.
When a tentative distance is successfully updated,
the \update{} function updates $\mu[\langle \cdot,\cdot \rangle]$ and $\mumax[\cdot]$ accordingly.
For example, if $\vcopy{v}{i}$ is successfully relaxed, any path distance involving $q_i$ can potentially be updated. 
Therefore, for $q_j\in \nq(q_i)$, we try to update $\mu[\langle q_i,q_j \rangle]$ by the path $q_i$--$v$--$q_j$,
which has a distance $\delta[\vcopy{v}{i}]+\delta[\vcopy{v}{j}]$.
If $\delta[\vcopy{v}{i}]+\delta[\vcopy{v}{j}]<\mu[\langle q_i,q_j \rangle]$,
the tentative distance on the edge $(q_i,q_j)$ is updated, 
thus the search radii $\mumax[i]$ and $\mumax[j]$ can potentially be reduced.
Then we update the radius $\mumax[i]$ by $\mumax[i]=\max\{\mu[\langle q_i,q_k \rangle]~|~q_k\in \nq(q_i)\}$ (similarly for $\mumax[j]$).

\hide{
\revise{
\myparagraph{Limitations of Our Algorithm.}
Compared to executing each query sequentially, our algorithm requires more memory since every vertex could have $|V_q|$ copies in the frontier,
which needs to allocate $|V_q|$ times more space for the frontier and distance array.
\xiaojun{We can remove the following discussions since it can be easily resolved.}
In addition, for better data locality, 
we assign an ID to the $i$-th copy of a vertex $u$ using $u\times n + i$ in the distance array.
Since all publicly available real-world graphs contain fewer than $2^{32}$ vertices, we use 32-bit integers for ids. 
However, for extremely large graphs such as COS5 in Table 3 (which has 321M vertices), this approach supports only up to 13 sources. 
To accommodate larger instances (more queries or larger graphs), 64-bit ids would be required.
}}

\hide{Compared to running the queries one at a time,
\ourlib{} can save some work by identifying searches from shared vertices,
and saturate parallelism for sparse or large-diameter graphs. }


\subsection{SSSP-Based Solutions}\label{sec:multi:sssp}
Besides the \bids{}-based solution mentioned above,
one simple way to deal with the batch is to run SSSP from all sources in the batch. 
We call it a \emph{plain} solution based on SSSP.
Although the idea is simple, it remains a reasonably effective solution for cases such as SSMT queries with a moderate number of targets.
In our experiments, even with five targets in an SSMT query, running SSSP on the source may
outperform other highly optimized solutions.

More generally, we are interested in knowing the minimum number of SSSP queries needed to handle the \querygraph{}.
Interestingly, for undirected graphs, this reduces to finding a subset of vertices in the \querygraph{} to cover all edges,
i.e., the \emp{vertex cover (VC)} problem.
Inspired by this, \ourlib{} includes a solution, which we call the \emp{VC} approach:
given the \querygraph{}, \ourlib{} first finds a vertex cover $V'\subseteq \vq$ of it, runs SSSP from
all vertices in $V'$, and combines the results.
We note that vertex cover on an arbitrary graph is NP-hard, and finding an optimal solution for a large \querygraph{} is infeasible.
\ourlib{} finds the VC by enumerating all possibilities for a small \querygraph{} and using a greedy algorithm when $\gq$ is large.
We experimentally verified the efficiency of the VC-based solution, particularly in cases 
where the number of edges in the query graph is much larger than the number of vertices in the vertex cover (VC).

\subsection{Takeaways and Further Discussions}

We present a performance comparison among several approaches for batch PPSP queries in \cref{sec:exp:batch-ppsp}, including
the two \bds{}-based solutions and two SSSP-based solutions.
Among multiple graphs and \querygraph{} patterns, the Multi-\bds{} achieves the best performance in most cases
by effectively leveraging shared information among the batch queries.

By abstracting the \querygraph{}, we can make use of \bds{} and VC and utilize shared information among the query pairs.
Such benefit is more significant on undirected graphs than directed graphs, 
but the idea can simply be extended to directed graphs. 
Note that on directed graphs, the query points will be separated into sources and targets, which effectively forms a bipartite graph.
The same algorithmic idea can just apply.
For Multi-\bds{}, we start searches from all vertices in the bipartite query graph, using the same idea mentioned in \cref{sec:multi:bids}.
The caveat is that we need to access both the outgoing edges and the incoming edges of a vertex, and some preprocessing is needed when reading the input graph.
For the VC-based approach, again assuming we can traverse the graph from both directions, the same idea based on vertex cover also directly applies.
Also, the optimal solution can be computed by bipartite graph matching in cubic work or better.
Then the searches from the selected vertices, either forward or backward, will cover all queries.

Although bidirectional \astar{} also has good performance in single PPSP queries, 
defining the heuristic functions becomes unclear when a search involves multiple targets.
We leave this as an intriguing future direction to explore. 

\myparagraph{Space Usage.}
For \bids-based approaches, the extra space needed for our Multi-\bds{} algorithm is $O(n\cdot|V_q|)$, 
where $|V_q|$ is the number of \emph{vertices} involved in the batch queries, 
instead of the number of queries in the batch.
For Plain-\bds{}, since we run the query one at a time, the extra space is $O(n)$.
In comparison, Plain$^*$-\bds{} requires $O(n\cdot|E_q|)$ extra space, as we need to maintain the distance array for each query.
For SSSP-based approaches, the extra space is upper-bounded by $O(n\cdot|V_q|)$ for both algorithms, since the algorithm needs to instantiate the searches from at most $|V_q|$ sources.

\myparagraph{Process Extremely Large Batches.}
\ourlib{} has a relatively high space usage if the given batches are very large.
Therefore, for applications that require a huge number (e.g., $\Omega(n)$) of PPSP queries on the same input graph, 
some other methodologies with preprocessing may be a better choice. 
We discuss these ideas in~\cref{sec:related}. 

\hide{
which can be large for a very large query batch.
In this case, we can simply process a subset of queries in turn, or use the SSSP-based solution mentioned below.
We note that when the batch is very large, say $O(n)$, then some shortest-path algorithms with preprocessing (see \cref{sec:related}) can be a better option.
} 

%% file: figures/fig-multi-ppsp.tex
\begin{figure}
  \centering
  \includegraphics[width=\columnwidth]{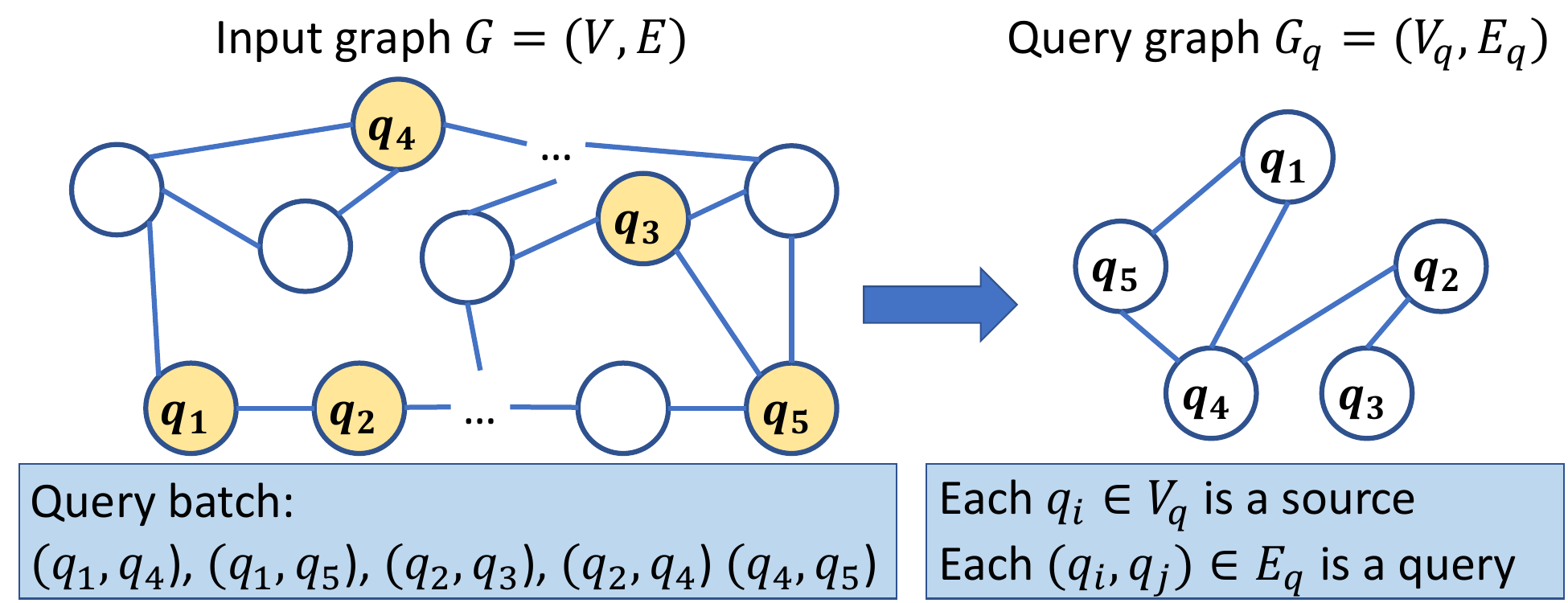}
  \caption{\textbf{Illustration for the query graph construction.}
  }\label{fig:multi-ppsp}
\end{figure}

%% file: imp.tex
\section{Implementation Optimizations}\label{sec:imp}

\myparagraph{Performance Improvement for \astar{} by Memoization.}
Although applying \astar{} to parallel PPSP is relatively easy,
there is still room for improvement empirically.
For instance, GraphIt~\cite{zhang2020optimizing} concludes that \astar{} is slower than a simple parallel PPSP with ET in many cases.
This is because computing the heuristic function itself could be expensive:
it usually requires computing the geometric distance between two multi-dimensional points,
which incurs extra cache misses in reading the coordinates.
In some cases, the heuristic function is given by spherical distances rather than Euclidean distances (e.g., in real-world road routing),
which can be even more expensive. 
\hide{
From our experiments, we found that preprocessing the heuristics and saving them in an array
can greatly reduce the running time in \astar{}.
However, since the heuristics need to be computed per query,
the preprocessing time can dominate the total running time if the given queries are very close.
}
This cost can be particularly expensive since $h(\cdot)$ is computed in the \prune{} function every time we relax
a vertex, 
and a vertex can be relaxed multiple times throughout the algorithm. 

Note that the heuristic function $h(\cdot)$ remains unchanged once $s$ and $t$ is selected. 
Hence, we can use an array to store the heuristics to avoid recomputation.
A potential challenge here is that we cannot afford to preprocess this value for \emph{all} vertices:
computing the full array takes $O(n)$ cost, but in an effective \astar{} execution,
the number of relaxed vertices should be fewer than $n$.
To overcome this challenge, we use memoization:
when we need the heuristic function for $v$ for the first time,
we compute it on the fly and save the result into the array.
For all subsequent queries that require the heuristic function on $v$,
it can be read directly from the array.
Our experiments (see \cref{fig:memoization}) show that memoization greatly improves the performance for \astar-based solutions.

Since our implementations are based on the stepping algorithms framework~\cite{dong2021efficient},
some optimizations are derived from them.
We discuss these optimizations in \ifconference{the full paper~\cite{dong2025parallelfull}}\iffullversion{\cref{app:imp}}.

%% file: exp.tex
\section{Experiments}\label{sec:exp}

\myparagraph{Setup.}
We run experiments on a 96-core machine (192 hyper-threads) with 4 $\times$ 2.1 GHz
Intel Xeon Gold 6252 CPUs (each with 36MB L3 cache) and 1.5TB of main memory.
Our code is implemented in C++, using ParlayLib~\cite{blelloch2020parlaylib}
to support fork-join scheduler and some parallel primitives.
We always use \texttt{numactl -i all} to interleave memory for parallel experiments.
For each source-target pair, the running time is determined by the average of the five runs after a warmup trial.
For each test, we select five source-target pairs in this setting and report the geometric mean for them. 
Running shortest paths on disconnected pairs can be optimized by a single pass of graph connectivity~\cite{dong2023provably,dhulipala2020connectit}.
Thus, we always select sources and targets from the largest connected component.

\input{tables/graphs.tex}
We test 14 graphs in four categories, including social networks, web graphs, road networks, and $k$-NN graphs.
The graph information is given in \cref{tab:graphs}.
For social and web graphs, we download them from the references.
Since there are no weights with the original graphs, we generate the weights uniformly at random in the range $[1, 2^{18}]$.
For road networks, we parse the longitude, latitude, and length of the road from OpenStreetMap~\cite{roadgraph}.
Since the coordinates are in the form of longitudes and latitudes, 
we compute the heuristic between two vertices using spherical distance.
For $k$-NN graphs, the original datasets are sets of low-dimensional points.
We use GeoGraph~\cite{wang2021geograph} to generate the corresponding $k$-NN graphs,
where each point is connected to its $k$ nearest neighbors.
We set $k=5$ for all $k$-NN graphs.
The heuristics for all $k$-NN graphs are Euclidean distances.
For both road networks and $k$-NN graphs, we use their original weights and coordinates.
We symmetrize the directed graphs. 
The selected graphs cover a wide range of sizes (from million to billion scale) and diameters (from 10 to $10^4$).
When taking the \emph{average} performance across multiple graphs, we always use the \emp{geometric mean}.

\ifconference{
\myparagraph{Additional Experiments.}
Due to page limits, we provide additional experimental results and analyses in the full version of the paper~\cite{dong2025parallelfull}, 
including comprehensive studies on distance percentiles, scalability, and memoization across all graphs, 
as well as experiments with larger batch sizes for batch PPSP queries.
}

\input{tables/PPSP-Astar.tex}

\subsection{Tested Algorithms}\label{sec:exp:algo}
\ourlib{} is implemented on top of the stepping algorithms framework~\cite{dong2021efficient}, which provides three implementations: $\rho$-stepping, $\Delta^*$-stepping, and Bellman-Ford.
We use $\Delta^*$-stepping (a variant of $\Delta$-stepping~\cite{meyer2003delta}) because it has the best performance on large-diameter graphs,
which are the primary use cases for \astar{} (road and $k$-NN graphs).
Note that \ourlib{} also supports $\rho$-stepping and Bellman-Ford,
and can easily be integrated with other SSSP algorithms.
To find the best parameter~$\Delta$ for each graph, we use the standard approach by starting with a small initial $\Delta$ to run the test
and doubling $\Delta$ until it converges to the minimum geometric mean running time.
For all other baselines using $\Delta$-stepping, we also find and use the best parameter $\Delta$ in the same way. 

\myparagraph{Baselines.}
We compare \ourlib{} with three SOTA baselines on single PPSP.
The first is the original parallel $\Delta^*$-stepping SSSP algorithm from SPAA'21~\cite{dong2021efficient},
which is also the base implementation we build upon. 
The other two are GraphIt from CGO'20~\cite{zhang2020optimizing} and Multi-Bucket Queue (MBQ) from SPAA'24~\cite{zhang2024multi}.
Both of them contain PPSP algorithms with early termination (ET) and \astar{} in their source code and experiments in their papers. 
We directly port them (denoted as GI-ET/GI-A*/MBQ-ET/MBQ-A*) in our experiments.
We select the best $\Delta$ for the two baselines as mentioned above,
and all other settings and parameters remain unchanged from the given sample scripts.
Since MBQ does not support floating-point distances due to bitmasking, 
we round floating-point values to integers in its experiments,
which favors MBQ over other implementations. 

For batched queries, we did not find publicly available implementations for general batch PPSP, 
so we compare four different approaches in \ourlib{} (all described in \cref{sec:framework}):
1) \textbf{Multi-\bids{}}: a \bds{}-based solution using Multi-\bds{},
2) \textbf{Plain-\bids{}}: a plain \bds{}-based solution running parallel \bds{} for each query one by one serially,
3) \textbf{Plain$^*$-\bids{}}: another plain \bds{}-based solution running parallel \bds{} for all queries simultaneously,
4) \textbf{VC-SSSP}: an SSSP-based solution using VC, and
5) \textbf{Plain-SSSP}: a plain solution running parallel SSSP from all sources.

\subsection{Single PPSP Query}\label{sec:exp:single}
\myparagraph{Overall Performance.}
We evaluate our PPSP algorithms with early termination (ET), bidirectional search (\bds), \astar, and bidirectional A$^*$ (\bdastar). 
We compare them with running SSSP from the source (SSSP),
the PPSP implementations in GraphIt~\cite{zhang2020optimizing}, and MBQ (GI-ET / GI-\astar{} / MBQ-ET / MBQ-\astar{}), 
both of which implement unidirectional search with early termination and \astar{}. 
The results are available in \cref{tab:single-query}, where each row represents an algorithm and each column represents a graph.
Note that the performance of ET, \bds{}, and \astar{}-based solutions is highly dependent on the distance between the two vertices.
Therefore, we evaluate our algorithms in three different distance percentiles (1-st, 50-th, and 99-th).
A query at the $x$-th distance percentile means the target is the $x$\% farthest vertex from the source. 
Since we do not have the coordinates on social and web graphs,
we only run \astar{} and \bdastar{} on road and $k$-NN graphs.

\myparagraph{Comparing with an SSSP Implementation.} We first compare our implementation with the straightforward baseline using SSSP. 
On social and web graphs, both ET and \bds{} always achieve better performance than SSSP
when the given queries are within the 50-th distance percentile.
However, when the distance percentile increases to the 99-th,
SSSP can achieve better performance than ET and \bds{}.
This is because, for the 99-th percentile setting,
both ET and \bds{} are unable to prune many vertices
and thus need to search almost all vertices in the graph.
Pruning only a few vertices cannot mitigate the extra overheads (e.g., maintaining $\mu$ and searching from two directions) compared to SSSP.

On road and $k$-NN graphs with geometric information, we can also compare \astar{} and \bdastar{}.
Comparing the average performance across all graphs, \bdastar{} achieves the \emph{best average performance on all distance percentiles}.
It achieves 76.0$\times$, 3.49$\times$, and 1.73$\times$ speedups compared to SSSP on 1-st, 50-th, and 99-th distance percentiles respectively.
If we only compare \bds{} and \astar{} to ET, we can see that both techniques can effectively reduce the running time in most test cases.
On average, across all graphs, their improvement on different percentiles can be 30--60\%. 
\bds{} is slightly more efficient than \astar{}.
The reason is that \astar{} requires to save the coordinates and compute the heuristics,
which can increase both the I/O and computation cost.
In contrast, \bds{} only needs to initiate two searches, and the other steps are almost the same as in ET. 
Adding distance heuristics by \astar{} is also very effective. 
\astar{} and \bdastar{} achieves up to 252$\times$ and 441$\times$ speedups compared to SSSP at the 1-st percentile, 
and still have up to 4.10$\times$ and 5.32$\times$ speedups at the 99-th percentile.

On road and \knn{} graphs, \bids{} and \astar{} still enable significant speedup over SSSP or ET in many cases, even at 99-percentile.
One reason might be that the SSSP algorithms on social and web graphs are highly optimized,
which leaves little space for further improvement with additional techniques.
For road and \knn{} graphs, which are sparser and inherently more difficult to achieve good parallelism~\cite{dong2021efficient},
using \bids{} or \astar{} still improves performance even within the 99\% closest vertices on many graphs.
In addition, \astar{} and \bdastar{} can further use the coordinates information to prioritize the search direction and improve the performance.
In general, the increasing distance still makes the advantage of \astar{} and \bdastar{} smaller,
and in certain cases, they can be slower than SSSP (i.e., \astar{} on AS and \bdastar{} on AF) on the 99-th percentile.
As discussed above, this is mainly due to the inadequacy of prunes and the extra overhead on the heuristics computations.

\myparagraph{Comparing with SOTA PPSP Implementations.}
We now compare \ourlib{} with GraphIt~\cite{zhang2020optimizing} and MBQ~\cite{zhang2024multi}. 
We first compare the standard PPSP query without geometric information (thus \astar{} is not applicable). 
GraphIt is almost always faster than MBQ, except for two cases,
and therefore we focus on comparing with GraphIt below. 
For ET, on social and web graphs, GI-ET has slightly better performance at the 1-st percentile,
while MBQ-ET and Ours-ET are comparable.
When the queried vertices are farther, Ours-ET becomes more efficient. 
On road and $k$-NN graphs, Ours-ET is almost always faster than GI-ET except for one graph, CH5,
where it remains competitive within 20\%.
On the geometric mean across all graphs, 
our ET is 1.33$\times$, 2.05$\times$, and 2.12$\times$ faster than GraphIt on the three distance percentiles, respectively. 
\ourlib{} further accelerates the queries by bidirectional search, leading to a final speedup of 2.8$\times$ compared to GI-ET on average
and 6.8$\times$ (or more, considering the timeout cases) faster compared to MBQ-ET.

We now compare the \astar{} performance across \ourlib{}, GraphIt, and MBQ. 
On the geometric mean across all road and $k$-NN graphs, 
Our \astar{} is 2.26$\times$, 2.52$\times$, and 2.59$\times$ faster than GraphIt on the three distance percentiles
and is 5.75$\times$, 4.63$\times$, and 4.28$\times$ faster than MBQ.
Part of the performance gain comes from memoization techniques, which reduce redundant heuristic computations.
With further integrating the bidirectional search, \ourlib{} finally achieves 4.4$\times$ speedup over GraphIt on average,
and is 6.2$\times$ faster than MBQ.

\input{figures/single-query.tex}
\input{figures/scalability_subset.tex}
\input{figures/fig-memoization.tex}

\myparagraph{The Study on Different Distance Percentiles.}
PPSP usually performs faster than SSSP.
However, the overheads of maintaining the pruning conditions and searching bidirectionally
become more visible when the queried pairs are farther.
Therefore, we provide an in-depth analysis of how the running time scales with fine-grained distance percentiles.
We pick a random source $s$ from the largest connected component
and select the targets based on the distances from $s$.
The first target is the 10-th closest vertex from $s$,
and we double the distance rank for every subsequent target (20-th, 40-th, etc.)
until we reach the farthest vertex.
We select a representative graph from each category and present the results in \cref{fig:distance}.
Results on all graphs are in \ifconference{the full paper~\cite{dong2025parallelfull}.}\iffullversion{\cref{fig:distance-full}.}

On social and web graphs, ET is not very effective.
Both ET and \bids{} can be slower than simply running SSSP when the queried vertices are far.
This is because these graphs have many high-degree vertices.
Once the PPSP search expands on these vertices, 
they can significantly increase the search space
and thus increase the running time.
Nevertheless, \bids{} still achieves about 10$\times$ speedup when the queried vertices are below the 1-st distance percentile.

\hide{
As discussed above, querying a farther target may involve some high-degree vertices in between
and thus substantially increase the running time.
On TW, as the distance percentile increases to 50\%,
ET has almost the same performance as in SSSP,
while \bds{} can still achieve a 92\% speedup.
When the distance percentile increases to 100\% (i.e., the target is the farthest vertex),
both ET and \bds{} have a 6\% slowdown over SSSP.
The result is similar on SD---PPSP has better performance until the 50-th distance percentile,
and is 5--7\% slower than SSSP when querying the farthest vertex.
}

On road and $k$-NN graphs, the running time grows almost linearly with increasing distance percentiles.
Both \bds{} and \astar{} can improve the performance compared to ET.
When the queried vertices are close (below 1-st percentile), \bds{} has the best performance in general.
When the distances are above the 1-st percentile threshold, \bdastar{} becomes the best.
This is because when the vertices are far enough, \bdastar{} can guide the search towards the midpoint more precisely
and thus can effectively prune more redundant searches in the wrong direction than \bids{}.

\myparagraph{Self-Relative Speedups.}
To measure parallelism, we present the self-relative scalability curves on one representative graph from each graph category in~\cref{fig:scalability}.
On the four representative graphs, all of our algorithms, as well as the SSSP implementation from the previous work,
have good scalability, achieving 31.0--67.8$\times$ speedups.
Interestingly, the simpler the algorithms are, the better scalability they have,
as plain algorithms have more work due to disoriented search
and thus can better saturate the processors with sufficient work.
Algorithms with more advanced techniques, such as \bids{}, skip many vertices with the heuristics,
which leads to less load balance and lower parallelism.
The full results on each graph are available in \iffullversion{\cref{fig:scalability-full}.}\ifconference{the full paper~\cite{dong2025parallelfull}.}

\myparagraph{Performance Study on Memoization.}
We test how memoization helps reduce the cost in \astar{}. 
We perform the same test as in \cref{tab:single-query} with and without memoization. 
In \cref{fig:memoization}, we show the average performance on road networks and \knn{} graphs, respectively.
All numbers are normalized to ET. 
The full results on each graph are presented in \iffullversion{\cref{fig:memoization-full}.}\ifconference{the full paper~\cite{dong2025parallelfull}.}

Without memoization, \astar{} and \bdastar{} can be slower than ET,
which is consistent with the findings in previous work.
Using memoization, the running time is reduced by 15\% on average.
Saving and reusing the heuristics greatly mitigate the overhead on the additional computations.
Memoization is more useful on road networks than in $k$-NN graphs,
which is due to the more complicated heuristic computations for spherical distances in road graphs than the Euclidean distances in $k$-NN graphs.
Memoization improves \astar{} and \bdastar{} by up to 1.74$\times$ on road graphs.

\input{figures/fig-heatmap.tex}

\subsection{Batch PPSP Queries}\label{sec:exp:batch-ppsp}

We discuss the tested batch PPSP algorithms in~\cref{sec:exp:algo} and present our experiments in~\cref{fig:multi-ppsp}.
We generate eight types of batch queries to simulate different use cases, as discussed in \cref{sec:intro,sec:framework}.
Their \querygraph{s} are outlined in \cref{fig:multi-ppsp}.
The numbers in \cref{fig:multi-ppsp} represent relative running time normalized to the fastest approach on each graph and each \querygraph{}.
To compare across queries,
we fix the number of sources to six
and vary only the queries between them.
The tested algorithms are discussed in \cref{sec:exp:algo}.
Note that since the vertices are randomly picked, the expected distance is at about the 50-th percentile.

Overall, our Multi-\bds{} \emph{achieves the best performance on about 80\% of the tests}.
Even in the worst case (star), Multi-\bds{} is within 50\% slower than the fastest baseline on average.
The speedups of Multi-\bds{} over other baselines become larger
when the \querygraph{} becomes denser (i.e., more sources are shared in the queries).
For example, on cliques, 
Multi-\bds{} is up to 2.81$\times$ faster (on CH5) even compared to the best baseline, and 2.20$\times$ on average.
Similarly, on other \querygraph{} patterns where each vertex may be involved in more queries,
such as chain, bipartite, and random, Multi-\bds{} have an overwhelming advantage over the others.

Even on separate queries (three disjoint $s$-$t$ pairs) where no sources are shared among different queries,
Multi-\bds{} still benefits from running all queries in parallel.
Compared to Plain-\bds{}, Multi-\bds{} achieves competitive performance on social and web graphs
and achieves up to 94\% speedups on road and $k$-NN graphs.
This indicates another benefit of using Multi-\bds{} on the batch:
SSSP/PPSP algorithms usually suffer from insufficient parallelism because
road and $k$-NN graphs are sparse and have large diameters.
Batching the queries can naturally saturate the processors with sufficient jobs.
Compared to Plain$^*$-\bds{}, Multi-\bds{} achieves a 14\% speedup on average,
which we believe is due to better cache access patterns from collocating the queries.
On average, Multi-\bds{} is 25\% faster than Plain-\bds{} and 14\% faster than Plain$^*$-\bds{}.

On other \querygraph{} patterns where the queries do not share vertex heavily, other approaches can be favorable.
For example, on star queries, only one SSSP is needed for SSSP-based algorithms,
while Multi-\bds{} needs to search from all six sources.
On social and web graphs, the advantage of running from only one source is obvious.
However, on road and $k$-NN graphs, Multi-\bds{} can still achieve comparable or even better performance (EU and CH5).
The reason is that on these graphs, \bds{} is already about 2.75-3.91$\times$ faster than SSSP with a single query (see \cref{tab:single-query}),
and running the queries in a batch also diminishes some redundant searches.

More generally, we can observe that a VC-based SSSP solution is almost always better (up to 2.55$\times$ times faster) 
than a plain SSSP solution
since the number of SSSPs running by VC is no more than the plain SSSP solution.
As expected, on graphs where a small number of vertices can cover all edges, such as star and fork graphs, 
the VC-based solution can even outperform the solution based on Multi-\bds{}, especially on social and web graphs.
As mentioned, this is because SSSP on these graphs is highly optimized,
which makes the additional work on maintaining \bds{} more significant.
Therefore, running one or two SSSPs may be cheaper than running \bids{} from all six vertices.

%% file: tables/graphs.tex
\begin{table}
  \small
  \centering

\setlength{\tabcolsep}{2.5pt}
\begin{tabular}{ccrrrrcl}
  &       & \multicolumn{1}{c}{\boldmath{}\textbf{$\boldsymbol{n}$}\unboldmath{}} & \multicolumn{1}{c}{\boldmath{}\textbf{$\boldsymbol{m}$}\unboldmath{}} & \multicolumn{1}{c}{\boldmath{}\textbf{$\boldsymbol{D}$}\unboldmath{}} & \multicolumn{1}{c}{\textbf{|LCC|\%}} & \textbf{Heuristics} & \multicolumn{1}{c}{\textbf{Notes}} \\
\midrule
\multirow{4}[2]{*}{\begin{sideways}\textbf{Social}\end{sideways}} & \textbf{OK} & 3.07M & 117M  & 9     & 100.0\% & -     & com-orkut~\cite{yang2015defining} \\
  & \textbf{LJ} & 4.85M & 42.9M & 19    & 99.9\% & -     & soc-LiveJournal1~\cite{backstrom2006group} \\
  & \textbf{TW} & 41.7M & 1.20B & 22    & 100.0\% & -     & Twitter~\cite{kwak2010twitter} \\
  & \textbf{FS} & 65.6M & 1.81B & 37    & 100.0\% & -     & Friendster~\cite{yang2015defining} \\
\midrule
\multirow{2}[2]{*}{\begin{sideways}\textbf{Web}\end{sideways}} & \textbf{IT} & 41.3M & 1.03B & 45    & 100.0\% & -     & it-2004~\cite{boldi2004webgraph} \\
  & \textbf{SD} & 89.2M & 1.94B & 35    & 98.8\% & -     & sd\_arc~\cite{webgraph} \\
\midrule
\multirow{4}[2]{*}{\begin{sideways}\textbf{Road}\end{sideways}} & \textbf{AF} & 33.5M & 44.5M & 3948  & 91.0\% & Spherical & Africa~\cite{roadgraph} \\
  & \textbf{NA} & 87.0M & 110M  & 4835  & 98.9\% & Spherical & North-America~\cite{roadgraph} \\
  & \textbf{AS} & 95.7M & 122M  & 8794  & 94.4\% & Spherical & Asia~\cite{roadgraph} \\
  & \textbf{EU} & 131M  & 166M  & 4410  & 98.9\% & Spherical & Europe~\cite{roadgraph} \\
\midrule
\multirow{4}[2]{*}{\begin{sideways}\boldmath{}\textbf{$k$-NN}\unboldmath{}\end{sideways}} & \textbf{HH5} & 2.05M & 6.50M & 1863  & 81.9\% & Euclidean & Household~\cite{uciml,wang2021geograph} \\
  & \textbf{CH5} & 4.21M & 14.8M & 14479 & 50.8\% & Euclidean & CHEM~\cite{fonollosa2015reservoir,wang2021geograph} \\
  & \textbf{GL5} & 24.9M & 78.7M & 21601 & 49.1\% & Euclidean & GeoLife~\cite{geolife,wang2021geograph} \\
  & \textbf{COS5} & 321M  & 979M  & 1180  & 99.8\% & Euclidean & Cosmo50~\cite{cosmo50,wang2021geograph} \\
\bottomrule
\end{tabular}%

\caption{\textbf{Graphs information.}
``$n$'' $=$ number of vertices.
``$m$'' $=$ number of edges.
``$D$'' $=$ diameter.
``|LCC|\%'' $=$ ratio of the largest connected components.}
\label{tab:graphs}

\end{table}%

%% file: tables/PPSP-Astar.tex
\begin{table*}
  \small
  \centering

\begin{tabular}{cl|cccc|cc|cccc|cccc|cc}
  &       & \multicolumn{4}{c|}{\textbf{Social}} & \multicolumn{2}{c|}{\textbf{Web}} & \multicolumn{4}{c|}{\textbf{Road}} & \multicolumn{4}{c|}{\boldmath{}\textbf{$k$-NN}\unboldmath{}} & \multicolumn{2}{c}{\cellcolor[rgb]{ 1,  1,  0}\textbf{MEAN}} \\
  &       & \textbf{OK} & \textbf{LJ} & \textbf{TW} & \textbf{FS} & \textbf{IT} & \textbf{SD} & \textbf{AF} & \textbf{NA} & \textbf{AS} & \textbf{EU} & \textbf{HH5} & \textbf{CH5} & \textbf{GL5} & \textbf{COS5} & \cellcolor[rgb]{ 1,  1,  0}\textbf{Heur.} & \cellcolor[rgb]{ 1,  1,  0}\textbf{ALL} \\
\midrule
\multirow{9}[2]{*}{\textbf{1-st}} & \textbf{SSSP~\cite{dong2021efficient}} & .092  & .053  & .964  & 2.35  & .463  & 1.89  & .700  & 1.17  & 1.56  & 1.37  & .032  & .243  & .800  & 1.95  & \cellcolor[rgb]{ 1,  1,  0}.618 & \cellcolor[rgb]{ 1,  1,  0}.545 \\
  & \textbf{Ours-ET} & .041  & .021  & .706  & .847  & .105  & .740  & .015  & .031  & .031  & .032  & .003  & .004  & .011  & .040  & \cellcolor[rgb]{ 1,  1,  0}.015 & \cellcolor[rgb]{ 1,  1,  0}.044 \\
  & \textbf{Ours-\bds} & \underline{.008} & \underline{.005} & \underline{.094} & .188  & .138  & \underline{.145} & .010  & .020  & .023  & .016  & \underline{.003} & .003  & .010  & .019  & \cellcolor[rgb]{ 1,  1,  0}.010 & \cellcolor[rgb]{ 1,  1,  0}\underline{.020} \\
  & \textbf{Ours-\astar} & -     & -     & -     & -     & -     & -     & .016  & .020  & .021  & .022  & .004  & .004  & .011  & .008  & \cellcolor[rgb]{ 1,  1,  0}.011 & \cellcolor[rgb]{ 1,  1,  0}- \\
  & \textbf{Ours-\bdastar} & -     & -     & -     & -     & -     & -     & \underline{.009} & \underline{.017} & \underline{.019} & \underline{.011} & .004  & .004  & \underline{.010} & \underline{.004} & \cellcolor[rgb]{ 1,  1,  0}\underline{.008} & \cellcolor[rgb]{ 1,  1,  0}- \\
  & \textbf{GI-ET~\cite{zhang2020optimizing}} & .027  & .016  & .417  & \underline{.167} & \underline{.062} & .401  & .029  & .057  & .062  & .052  & .116  & \underline{.003} & .070  & .058  & \cellcolor[rgb]{ 1,  1,  0}.041 & \cellcolor[rgb]{ 1,  1,  0}.058 \\
  & \textbf{GI-\astar~\cite{zhang2020optimizing}} & -     & -     & -     & -     & -     & -     & .017  & .032  & .032  & .031  & .081  & .003  & .079  & .013  & \cellcolor[rgb]{ 1,  1,  0}.025 & \cellcolor[rgb]{ 1,  1,  0}- \\
  & \textbf{MBQ-ET~\cite{zhang2024multi}} & .037  & .036  & 1.21  & .393  & .340  & .834  & .206  & .392  & .429  & .484  & .035  & .025  & .211  & .178  & \cellcolor[rgb]{ 1,  1,  0}.165 & \cellcolor[rgb]{ 1,  1,  0}- \\
  & \textbf{MBQ-\astar~\cite{zhang2024multi}} & -     & -     & -     & -     & -     & -     & .069  & .111  & .119  & .102  & .031  & .022  & .142  & .030  & \cellcolor[rgb]{ 1,  1,  0}.064 & \cellcolor[rgb]{ 1,  1,  0}- \\
\midrule
\multirow{9}[2]{*}{\textbf{50-th}} & \textbf{SSSP~\cite{dong2021efficient}} & .091  & .052  & .972  & 2.36  & .461  & 1.88  & .701  & 1.17  & 1.57  & 1.36  & .031  & .245  & .795  & 1.92  & \cellcolor[rgb]{ 1,  1,  0}.615 & \cellcolor[rgb]{ 1,  1,  0}.543 \\
  & \textbf{Ours-ET} & .077  & .049  & .878  & 1.94  & .314  & 1.60  & .388  & .538  & .783  & .620  & .024  & .121  & \underline{.291} & 1.08  & \cellcolor[rgb]{ 1,  1,  0}.313 & \cellcolor[rgb]{ 1,  1,  0}.341 \\
  & \textbf{Ours-\bds} & \underline{.075} & \underline{.033} & \underline{.805} & \underline{1.30} & \underline{.223} & \underline{1.08} & \underline{.304} & .315  & .498  & .348  & \underline{.012} & \underline{.089} & .352  & .515  & \cellcolor[rgb]{ 1,  1,  0}.206 & \cellcolor[rgb]{ 1,  1,  0}\underline{.239} \\
  & \textbf{Ours-\astar} & -     & -     & -     & -     & -     & -     & .342  & .241  & \underline{.479} & .303  & .025  & .140  & .416  & .132  & \cellcolor[rgb]{ 1,  1,  0}.197 & \cellcolor[rgb]{ 1,  1,  0}- \\
  & \textbf{Ours-\bdastar} & -     & -     & -     & -     & -     & -     & .426  & \underline{.215} & .570  & \underline{.243} & .014  & .105  & .472  & \underline{.110} & \cellcolor[rgb]{ 1,  1,  0}\underline{.177} & \cellcolor[rgb]{ 1,  1,  0}- \\
  & \textbf{GI-ET~\cite{zhang2020optimizing}} & .149  & .070  & 2.68  & 1.64  & .609  & 3.22  & 1.20  & 1.06  & 1.61  & 1.07  & .114  & .094  & 1.59  & 2.07  & \cellcolor[rgb]{ 1,  1,  0}.725 & \cellcolor[rgb]{ 1,  1,  0}.701 \\
  & \textbf{GI-\astar~\cite{zhang2020optimizing}} & -     & -     & -     & -     & -     & -     & .993  & .479  & 1.10  & .626  & .148  & .155  & 1.48  & .333  & \cellcolor[rgb]{ 1,  1,  0}.497 & \cellcolor[rgb]{ 1,  1,  0}- \\
  & \textbf{MBQ-ET~\cite{zhang2024multi}} & .195  & .191  & 2.34  & 3.85  & 1.18  & 3.45  & t.o.  & t.o.  & t.o.  & t.o.  & .288  & .852  & t.o.  & t.o.  & \cellcolor[rgb]{ 1,  1,  0}- & \cellcolor[rgb]{ 1,  1,  0}- \\
  & \textbf{MBQ-\astar~\cite{zhang2024multi}} & -     & -     & -     & -     & -     & -     & 1.49  & .772  & 1.65  & .727  & .148  & .684  & 6.50  & .529  & \cellcolor[rgb]{ 1,  1,  0}.912 & \cellcolor[rgb]{ 1,  1,  0}- \\
\midrule
\multirow{9}[1]{*}{\textbf{99-th}} & \textbf{SSSP~\cite{dong2021efficient}} & \underline{.090} & \underline{.052} & \underline{.968} & \underline{2.35} & .454  & 1.88  & .704  & 1.17  & 1.57  & 1.35  & .032  & .246  & .797  & 1.94  & \cellcolor[rgb]{ 1,  1,  0}.618 & \cellcolor[rgb]{ 1,  1,  0}.543 \\
  & \textbf{Ours-ET} & .092  & .057  & 1.01  & 2.72  & .448  & 1.92  & .730  & 1.18  & 1.66  & 1.41  & .035  & .244  & .856  & 2.14  & \cellcolor[rgb]{ 1,  1,  0}.648 & \cellcolor[rgb]{ 1,  1,  0}.570 \\
  & \textbf{Ours-\bds} & .095  & .062  & 1.08  & 2.48  & \underline{.416} & \underline{1.77} & \underline{.505} & \underline{.795} & \underline{.961} & \underline{.982} & .030  & \underline{.157} & .535  & 1.14  & \cellcolor[rgb]{ 1,  1,  0}.426 & \cellcolor[rgb]{ 1,  1,  0}\underline{.447} \\
  & \textbf{Ours-\astar} & -     & -     & -     & -     & -     & -     & .704  & 1.04  & 1.59  & 1.21  & .034  & .288  & .387  & .473  & \cellcolor[rgb]{ 1,  1,  0}.473 & \cellcolor[rgb]{ 1,  1,  0}- \\
  & \textbf{Ours-\bdastar} & -     & -     & -     & -     & -     & -     & .858  & .976  & 1.36  & 1.05  & \underline{.021} & .191  & \underline{.152} & \underline{.364} & \cellcolor[rgb]{ 1,  1,  0}\underline{.356} & \cellcolor[rgb]{ 1,  1,  0}- \\
  & \textbf{GI-ET~\cite{zhang2020optimizing}} & .233  & .104  & 3.37  & 2.59  & .739  & 4.38  & 2.17  & 2.30  & 3.51  & 2.46  & .121  & .189  & 5.36  & 3.94  & \cellcolor[rgb]{ 1,  1,  0}1.46 & \cellcolor[rgb]{ 1,  1,  0}1.21 \\
  & \textbf{GI-\astar~\cite{zhang2020optimizing}} & -     & -     & -     & -     & -     & -     & 2.11  & 2.15  & 3.73  & 2.35  & .179  & .298  & 1.75  & 1.34  & \cellcolor[rgb]{ 1,  1,  0}1.22 & \cellcolor[rgb]{ 1,  1,  0}- \\
  & \textbf{MBQ-ET~\cite{zhang2024multi}} & .311  & .245  & 4.27  & 6.38  & 1.43  & 4.94  & t.o.  & t.o.  & t.o.  & t.o.  & .589  & 2.27  & t.o.  & t.o.  & \cellcolor[rgb]{ 1,  1,  0}- & \cellcolor[rgb]{ 1,  1,  0}- \\
  & \textbf{MBQ-\astar~\cite{zhang2024multi}} & -     & -     & -     & -     & -     & -     & 3.12  & 2.66  & 5.02  & 2.30  & .177  & 1.61  & 7.37  & 1.38  & \cellcolor[rgb]{ 1,  1,  0}2.02 & \cellcolor[rgb]{ 1,  1,  0}- \\
\end{tabular}%

\caption{\textbf{Running time (in seconds) on five single-SSSP algorithms.}
Smaller is better. 
$i$-th means the targets are the $i$\% farthest vertex from the sources.
The fastest algorithm on each graph and each percentile is underlined.
The full information of the graphs is given in~\cref{tab:graphs}.
``SSSP'' $=$ Single-source shortest paths~\cite{dong2021efficient}.
``ET'' $=$ Early termination.
``\bds'' $=$ Bidirectional search.
``\astar'' $=$ \astar{} search.
``\bdastar'' $=$ Bidirectional \astar{} search.
``GI'' $=$ GraphIt~\cite{zhang2020optimizing}.
``MBQ'' $=$ Multi Bucket Queues~\cite{zhang2024multi}.
``Mean'' $=$ Geometric means across graphs with heuristics (road and $k$-NN graphs) and all graphs.
``Heur.'' $=$ Heuristic-based graphs.
``-'' $=$ Not applicable.
``t.o.'' $=$ Timeout (exceeding ten seconds).
\astar{} search is not applicable to social and web graphs as they do not have heuristics (coordinates).
}\label{tab:single-query}
\vspace{-1em}
\end{table*}

%% file: figures/single-query.tex
\begin{figure*}
  \small
  \centering
  \includegraphics[width=\textwidth]{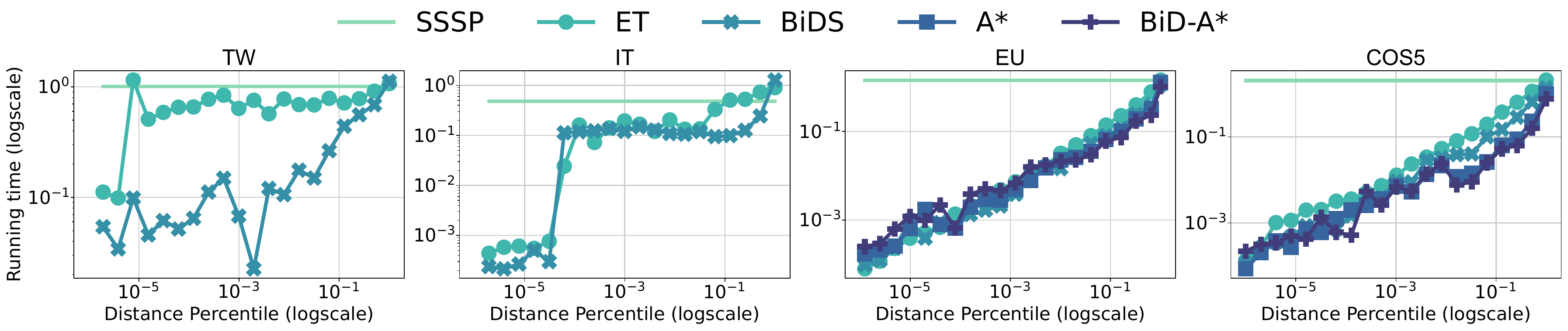}
  \caption{\textbf{Running time vs. distance percentile for single PPSP algorithms on four representative graphs from each category.} Lower is better.}
  \label{fig:distance}
\end{figure*}

%% file: figures/scalability_subset.tex
\begin{figure*}
  \small
  \centering
  \includegraphics[width=\textwidth]{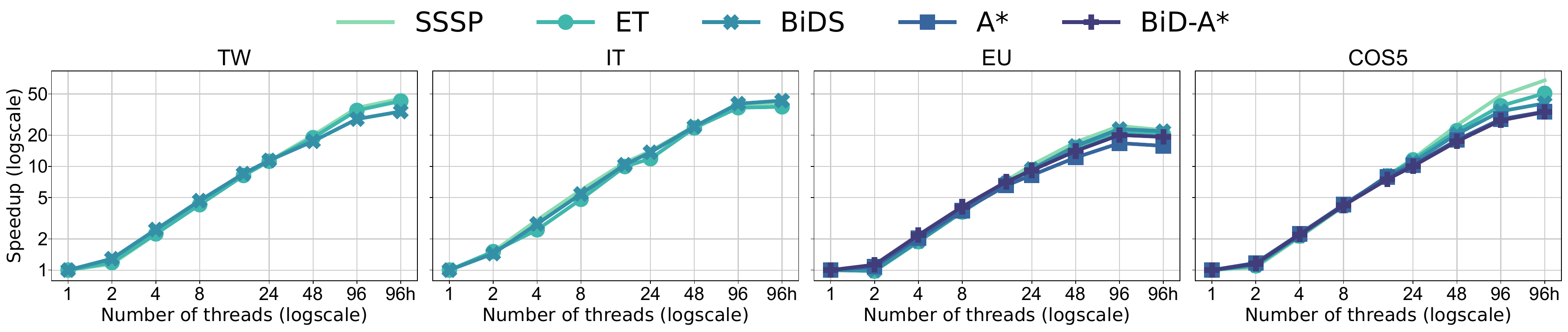}
  \caption{\textbf{Self-relative speedups of our single PPSP algorithms on four representative graphs from each category.} Higher is better.}
  \label{fig:scalability}
\end{figure*}

%% file: figures/fig-memoization.tex
\begin{figure}
  \small
  \centering
  \includegraphics[width=\columnwidth]{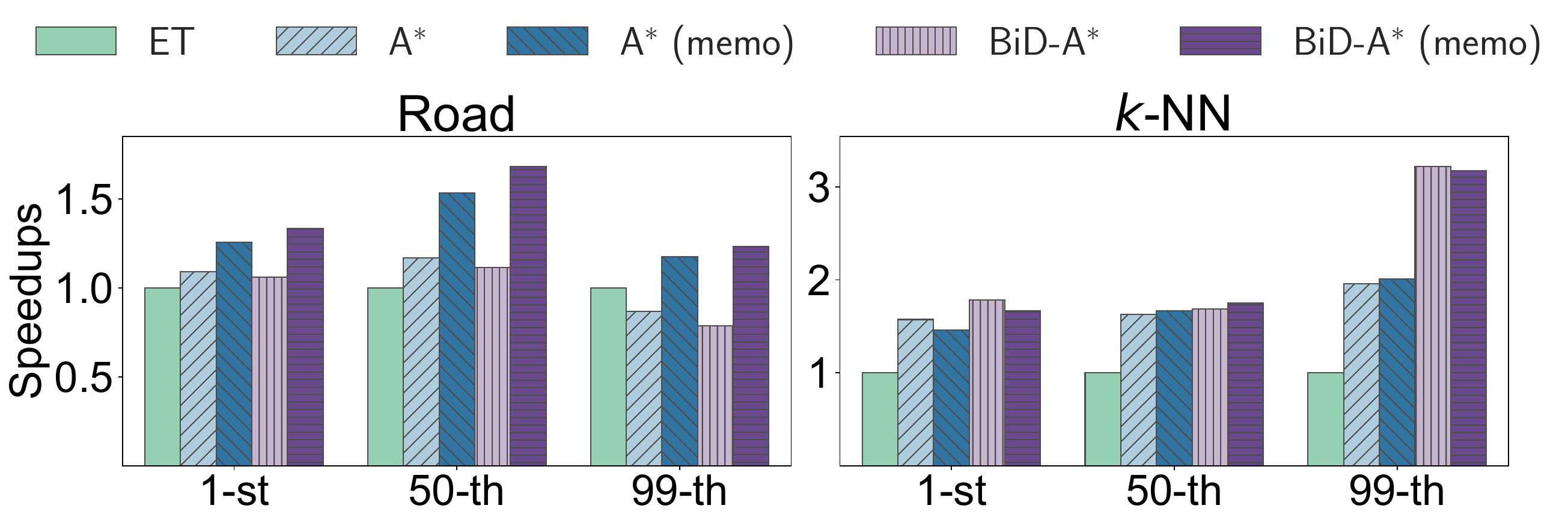}
  \caption{\textbf{Relative performance over ET with and without memoization on road and $k$-NN graphs.} 
  Numbers are running times normalized to ET. 
  ET is always one.
  Higher is better.
  We take the geometric mean on all road and $k$-NN graphs.
  ``memo'' denotes memoization.
  \label{fig:memoization}}
\end{figure}

%% file: figures/fig-heatmap.tex
\begin{figure*}
  \centering
  \includegraphics[width=\textwidth]{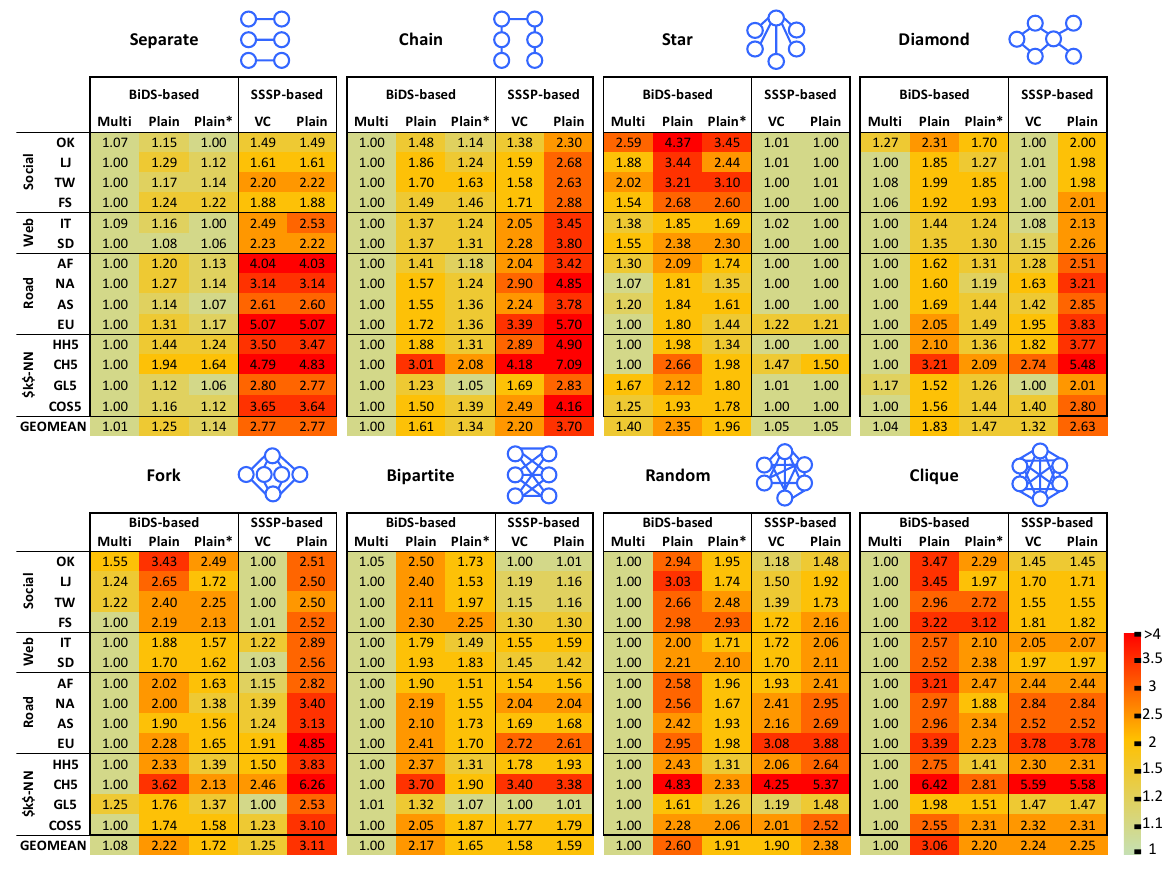}
  \caption{\textbf{Heatmap on batch PPSP queries.}
  The numbers are the running times normalized to the fastest algorithm on each test instance.
  Smaller or green is better.
  ``GEOMEAN'' $=$ geometric means on all graphs.
  We run two \bds{}-based algorithms and two SSSP-based algorithms as introduced in \cref{sec:framework}.
  For \bds{}-based algorithms, Multi is to evaluate all queries in a batch,
  Plain is to runs each query with our parallel BiDS algorithm one at a time,
  and Plain$^*$ is to run all queries simultaneously in parallel, and each query is a parallel BiDS algorithm.
  For SSSP-based algorithms, VC is to run the SSSP algorithm from the minimum required sources induced by the vertex cover of the \querygraph,
  and Plain is to run the SSSP algorithm from all given sources.}
  \label{fig:multi-ppsp}
\end{figure*}

%% file: related-work.tex
\section{Related Work}\label{sec:related}

\myparagraph{Existing Parallel Algorithms for \bids{} and Batch PPSP.}
Although \bids{} is shown to be effective in accelerating sequential PPSP queries, 
we are unaware of any fully parallelized \bids{} algorithm throughout the query phase.
Some existing works~\cite{rios2011pnba,vaira2011parallel} parallelize bidirectional Dijkstra with two threads,
one for the forward search and one for the backward search,
which naturally does not scale to more than two threads.
PnP~\cite{xu2019pnp} 
uses bidirectional searches in preprocessing to predict the more efficient direction (from either the source or the target) to perform the search. 
Once predicted, it continues the search in only one direction, either forward or backward.

There have also been studies on batch PPSP queries~\cite{knopp2007computing,geisberger2010engineering,mazloumi2019multilyra}.
Some of them~\cite{knopp2007computing,geisberger2010engineering} support only a specific type of query, such as SSMT or subset APSP, and do not generalize to arbitrary batches.
Multilyra~\cite{mazloumi2019multilyra} studies batch PPSP queries in the distributed system setting,
with the goal of processing a batch of queries at a time and thus amortizing the communication cost.
However, it does not consider the connections between the given queries.

\myparagraph{Other Algorithms for Computing Shortest Paths.}
We note that there is a large body of work on efficiently computing shortest paths (mostly in the sequential setting), and we refer the audience to an excellent survey~\cite{bast2016route}.
Most other techniques require some preprocessing to accelerate query time.
Some famous ones include contraction hierarchies (CH)~\cite{geisberger2012exact}, transit nodes routing~\cite{bast2007fast}, pruned landmark labeling (PLL)~\cite{akiba2013fast}, and ALT (to accelerate \astar{})~\cite{goldberg2005computing}.
Most of these ideas are orthogonal to the techniques studied in this paper: \bds{}, \astar{}, and batch PPSP, which do not require preprocessing.
We should note that preprocessing in shortest-path algorithms is double-edged---while queries can be significantly accelerated, the preprocessing can also take much time, and sometimes much more space as well (e.g., PLL).
The preprocessing-free methodology studied in this paper is potentially preferred in certain scenarios, such as
when fewer total queries are performed, graphs are larger (so less auxiliary space is available), 
and/or graphs change frequently. 


%% file: concl.tex
\section{Conclusions}\label{sec:concl} 

\hide{Due to the lack of strict visiting order (Dijkstra's order or best-first order), it has remained a long-standing challenge in how to integrating the algorithmic techniques such as bidirectional search (\bds{}) and \astar{} to the parallel setting.}
In this paper, we systematically study parallel point-to-point shortest paths (PPSP). 
We propose \ourlib{}, a parallel library for PPSP with multiple algorithms for both single and batch PPSP queries. 
Algorithmically, we propose a list of techniques to 
integrate bidirectional search (\bds{}) and \astar{} into existing parallel SSSP algorithms to achieve simple implementation with high parallelism. 
\ourlib{} supports \bds{}, \astar{}, and their combination \bdastar{}. 
We unify the methodologies for all of them in a PPSP framework in~\cref{alg:ppsp} by plugging in three user-defined functions.
Using the framework, all these algorithms are simple, easy to program, and have high performance. 

For batch queries, we give a novel abstraction based on the \querygraph{}. 
This abstraction allows the algorithmic ideas to be studied independently of the query types. 
We show how to use vertex cover and our \bds{} technique to support arbitrary query batches, 
which can be especially efficient when the queries share vertices.  
Our Multi-\bds{} for batch PPSP can also be implemented by our PPSP framework in a simple way. 

We conducted a thorough experimental study and showed that with the careful algorithmic design, both \bids{} and \astar{} indeed lead to significant improvements in the parallel setting. 
Namely, \bds{} and \astar{} can consistently improve the PPSP performance, and \bdastar{}, which combines their advantages, 
usually gives the best overall performance, which aligns with the observations in the sequential setting. 
Our \bids{}-based solutions outperform state-of-the-art unidirectional PPSP solutions such as GraphIt~\cite{zhang2020optimizing} and MBQ~\cite{zhang2024multi}. 
In the batch PPSP setting, our proposed techniques effectively reuse shared information in the batch and consistently achieve better performance than the plain implementations. 

%% file: ppsp.bbl

%% file: appendix.tex
\section{Proof for \cref{thm:bdastar-correctness}}

We present the proof for \cref{thm:bdastar-correctness}, which states that our \bdastar{} algorithm can correctly compute the shortest distance between the source $s$ and target $t$. 
\vspace{10pt}

\begin{proof}
  Recall that $\delta[v]$ is the tentative distance and finally converges to the true distance $d(v)$, so we have $\delta[v]\geq d(v)$.
  Assume the true distance between $s$ and $t$ is $\mu'=d(s,t)$.
  Note that when $\mu$ has reached $\mu'$ (i.e., $\mu=\mu'$) in the algorithm, then the shortest distance $d(s,t)$ has been computed.

  Hence, we now focus on the case when $\mu'<\mu$.
  As in the proof of \cref{thm:bids-correctness}, assume one of the shortest paths between $s$ and $t$ is
  $P=\{v_0=s, v_1, \ldots, v_k=t\}$.
  For every vertex $v_i\in P$, let
  \[
  \begin{aligned}
    F(v_i) &= d(s,v_i)+h(v_i), \\
    B(v_i) &= d(v_i,t)-h(v_i).
  \end{aligned}
  \]
  which are the true-distance analogues of the forward and backward pruning terms.
  Clearly, for all $v_i\in P$, we have
  \[
    F(v_i)+B(v_i)=d(s,v_i)+d(v_i,t)=\mu'.
  \]

  By consistency of $h(\cdot)$, for every edge $(v_i,v_{i+1})$ on $P$,
  \[
    h(v_i)\leq w(v_i,v_{i+1})+h(v_{i+1}),
  \]
  and therefore
  \[
    F(v_i)=d(s,v_i)+h(v_i)\leq d(s,v_{i+1})+h(v_{i+1})=F(v_{i+1}).
  \]
  Similarly,
  \[
    B(v_i)=d(v_i,t)-h(v_i)\geq d(v_{i+1},t)-h(v_{i+1})=B(v_{i+1}).
  \]
  Thus, along the shortest path $P$, the forward term $F(\cdot)$ is nondecreasing and the backward term $B(\cdot)$ is nonincreasing.

  Let $m$ be the largest index such that $F(v_m)<\mu/2$.
  Under a general consistent heuristic, this index may be $-1$ if no such vertex exists, or $k$ if all vertices satisfy the condition.
  Then, for every $0\leq i\leq m$, we have $F(v_i)\leq F(v_m)<\mu/2$, so these vertices are not pruned by the forward search once their true distances are obtained.
  For every $m+1\leq i\leq k$, we have $F(v_i)\geq \mu/2$, and thus
  \[
    B(v_i)=\mu'-F(v_i)\leq \mu'-\mu/2<\mu/2,
  \]
  so these vertices are not pruned by the backward search once their true distances are obtained.

  We now prove that the true distances on the path needed to recover $\mu'$ will be computed.
  In the forward search, $s=v_0$ is the source with $\delta[\vcopy{s}{+}]=d(s,s)=0$.
  If $m\geq 0$, then by induction on $i$, each vertex $v_i$ for $0\leq i\leq m$ is eventually relaxed from $v_{i-1}$ to its true distance $d(s,v_i)$.
  Since $F(v_i)<\mu/2$, it is not pruned when this true distance is obtained, and so the forward search can continue along the path.
  Similarly, if $m+1\leq k$, then by induction from $t=v_k$ backward, each vertex $v_i$ for $m+1\leq i\leq k$ is eventually relaxed to its true distance $d(v_i,t)$ in the backward search.
  Since $B(v_i)<\mu/2$, it is not pruned when this true distance is obtained, and so the backward search can continue along the path.

  If $m=k$, then the forward search reaches $t$ and the $\update{}$ function sets
  $\mu=\delta[\vcopy{t}{+}]+\delta[\vcopy{t}{-}]=d(s,t)=\mu'$.
  Symmetrically, if $m=-1$, then the backward search reaches $s$ and updates $\mu$ to $\mu'$.
  Otherwise, the forward search reaches $v_m$ with the true distance $d(s,v_m)$, and the backward search reaches $v_{m+1}$ with the true distance $d(v_{m+1},t)$.
  Since $(v_m,v_{m+1})$ is an edge on the shortest path, eventually one of these two searches relaxes across this edge after the other direction has already obtained its true distance.
  When this later relaxation happens, the $\update{}$ function uses either
  \[
    \delta[\vcopy{v_m}{+}]+\delta[\vcopy{v_m}{-}]=d(s,v_m)+d(v_m,t)=\mu'
  \]
  or
  \[
    \delta[\vcopy{v_{m+1}}{+}]+\delta[\vcopy{v_{m+1}}{-}]=d(s,v_{m+1})+d(v_{m+1},t)=\mu'.
  \]
  Hence, the return value $\mu$ is eventually updated to the true distance $\mu'$.

  Therefore, our \bdastar{} algorithm correctly computes the shortest distance between $s$ and $t$.
\end{proof}

\section{Other Implementation Optimizations}\label{app:imp}
\myparagraph{Sparse-Dense Optimization.}
We use the sparse-dense optimization as in many other systems~\cite{shun2013ligra,beamer2015gap}.
When the frontier size is small (sparse mode), we use a parallel hash bag~\cite{wang2023parallel,dong2024pasgal}
to maintain the frontier as in existing work~\cite{liu2025parallel,wan2025parallel}, 
which is a parallel data structure for a set that supports inserting and resizing. 
Otherwise, when the frontier size reaches a constant fraction of $n$ (dense mode),
we use a boolean array of size $n$ to indicate if a vertex is in the frontier or not.
\hide{Although we need to check an array of size $n$ in dense mode,
since the frontier size is already a constant fraction of $n$,
it does not change the overall complexity.}
In practice, it can be more efficient than using a parallel hash bag when the frontier is large
because maintaining an array of boolean flags incurs much smaller overhead and is more cache-friendly.

\myparagraph{Bidirectional Relaxation Optimization.}
The direction optimization is introduced in~\cite{dong2021efficient} for SSSP.
Similar ideas are also studied for BFS (referred to as the pull-push optimization)~\cite{Beamer12}.
If the input graph is undirected, when the algorithm is about to relax the neighbors of a vertex $u$,
it first relaxes $u$ using its neighbors $v\in N(u)$ (i.e., check $\delta[v]+w(v, u) < \delta[u]$).
It aims to get the best possible tentative distance for $u$ before it relaxes other vertices.
Since the neighbor lists are likely to reside in the cache,
doing the bidirectional relaxation only incurs little overhead.

\myparagraph{Optimization for Disconnected Queries.}
In the extreme case that the given query $(s,t)$ is disconnected,
\bds{} needs to search both connected components that $s$ and $t$ are in,
which can be much slower than a unidirectional search
if the sizes of the two connected components differ significantly.
In order to reduce the overhead as much as possible,
our algorithm checks the direction of the vertices in the frontier before each round.
If all the vertices are from either $s$ or $t$ and $\mu$ is never updated,
the algorithm can terminate right away, returning $\infty$ as the distance.

\section{More Experimental Results}\label{app:exp}
For completeness, we present the full experimental results on each graph in this section.

\myparagraph{Large Batch Sizes on Batch PPSP.}
We evaluate our batch PPSP algorithms on larger query batches in~\cref{fig:multi-ppsp-large}.
We fix the number of query sources to 30 and increase the number of query pairs from 100 to 300.
Overall, our Multi-BiDS algorithm consistently achieves the best geometric mean running time,
outperforming the next best baseline by up to 1.32$\times$.
As the number of query pairs increases, the performance gap between Multi-BiDS and the other algorithms becomes even more pronounced.

\myparagraph{In-Depth Analysis on Different Distance Percentiles.}
We study the relationship between the running time and the distance between the source and target
and presented the full results in~\cref{fig:distance-full}.
The conclusion is similar to the representative graphs in the main paper.

\myparagraph{Self-Relative Speedups.}
We presented the scalability curves on each graph in~\cref{fig:scalability-full}.
As discussed in the main paper, using more advanced techniques will relatively reduce the amount of work and increase load imbalance,
as some of the vertices in the frontier are pruned from the searches.
Therefore, plain techniques can have better speedups.
On most of the graphs, our algorithms achieve good scalability.
The only outlier is CH5, on which the speedups are less than ten.
This is because the edge weights are skewed on CH5, which makes the original SSSP algorithms not scalable.
In future work, we plan to study shortest-path algorithms on graphs with skewed weights.

\myparagraph{Performance Study on Memoization.}
We presented the full results on the memoization technique in~\cref{fig:memoization-full}.
Using memoization on road graphs is more beneficial than on $k$-NN graphs, 
as road graphs use spherical distance as a heuristic,
which is more complex to compute than the Euclidean distance used in $k$-NN graphs.

\input{figures/fig-large-batch.tex}
\input{figures/distance_full.tex}
\input{figures/scalability_full.tex}
\input{figures/memoization_full.tex}

%% file: figures/fig-large-batch.tex
\begin{figure*}[htp]
  \centering
  \includegraphics[width=\textwidth]{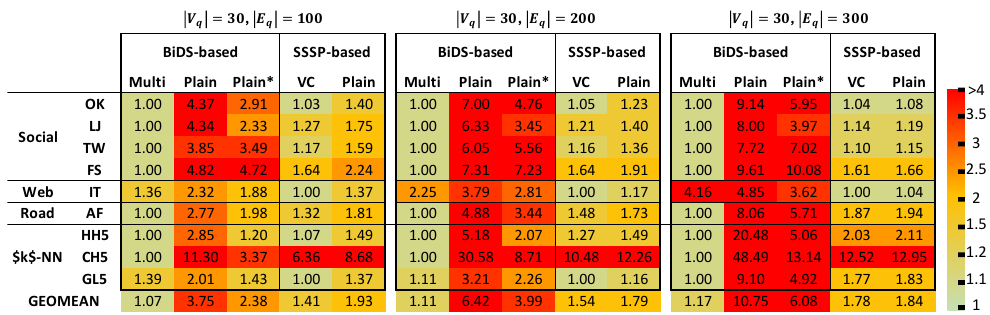}
  \caption{\textbf{Heatmap on large batch PPSP queries.}
  The numbers are the running times normalized to the fastest algorithm on each test instance.
  Smaller or green is better.
  ``GEOMEAN'' $=$ geometric means on all graphs.
  We run two \bds{}-based algorithms and two SSSP-based algorithms as introduced in \cref{sec:framework}.
  For \bds{}-based algorithms, Multi is to evaluate all queries in a batch,
  Plain is to runs each query with our parallel BiDS algorithm one at a time,
  and Plain$^*$ is to run all queries simultaneously in parallel, and each query is a parallel BiDS algorithm.
  For SSSP-based algorithms, VC is to run the SSSP algorithm from the minimum required sources induced by the vertex cover of the \querygraph,
  and Plain is to run the SSSP algorithm from all given sources.}
  \label{fig:multi-ppsp-large}
\end{figure*}

%% file: figures/distance_full.tex
\begin{figure*}
  \small
  \centering
  \includegraphics[width=\textwidth]{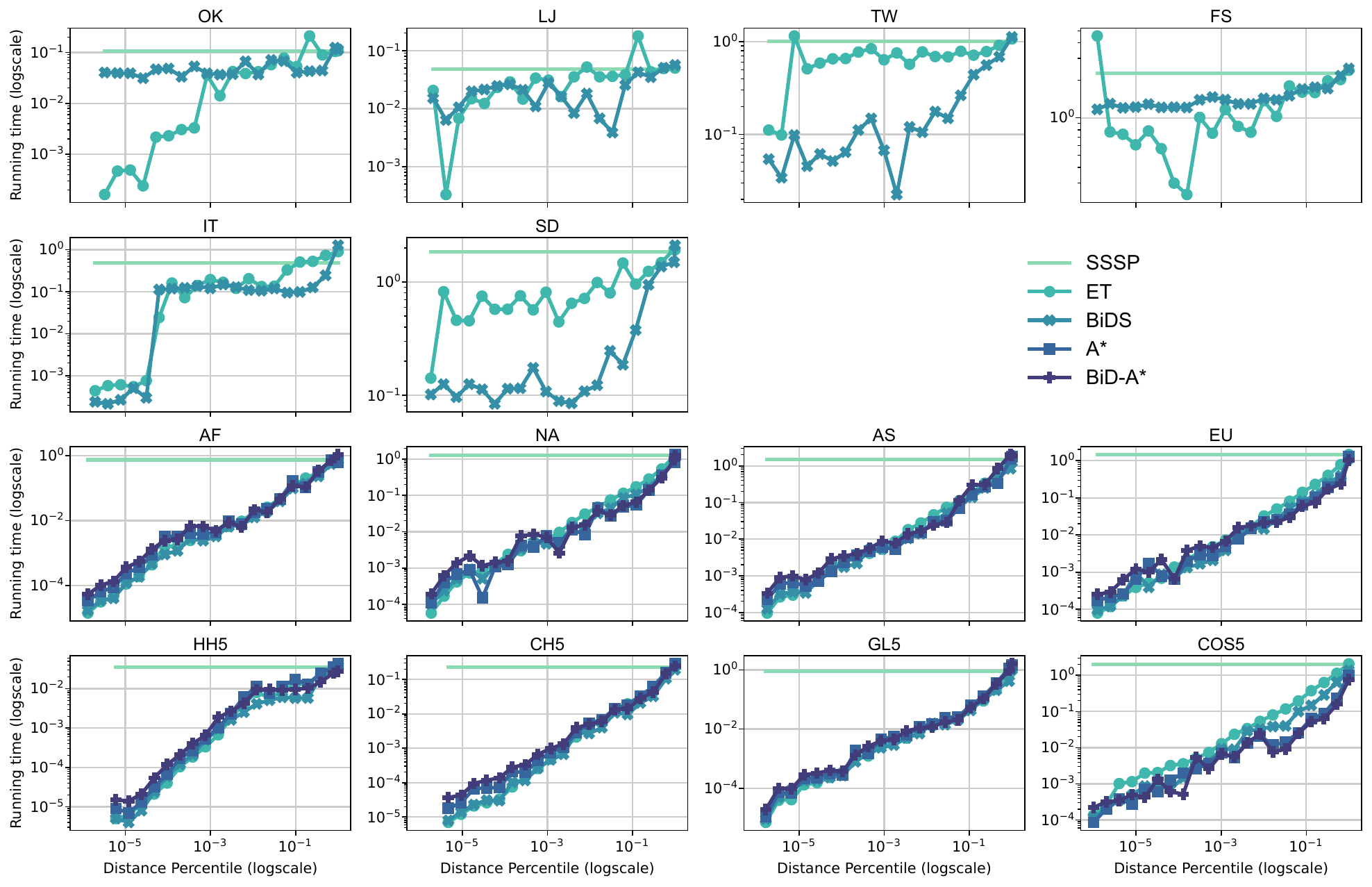}
  \caption{\textbf{Running time vs. distance percentile for single PPSP algorithms on each graph.} Lower is better.}
  \label{fig:distance-full}
\end{figure*}

%% file: figures/scalability_full.tex
\begin{figure*}
  \small
  \centering
  \includegraphics[width=\textwidth]{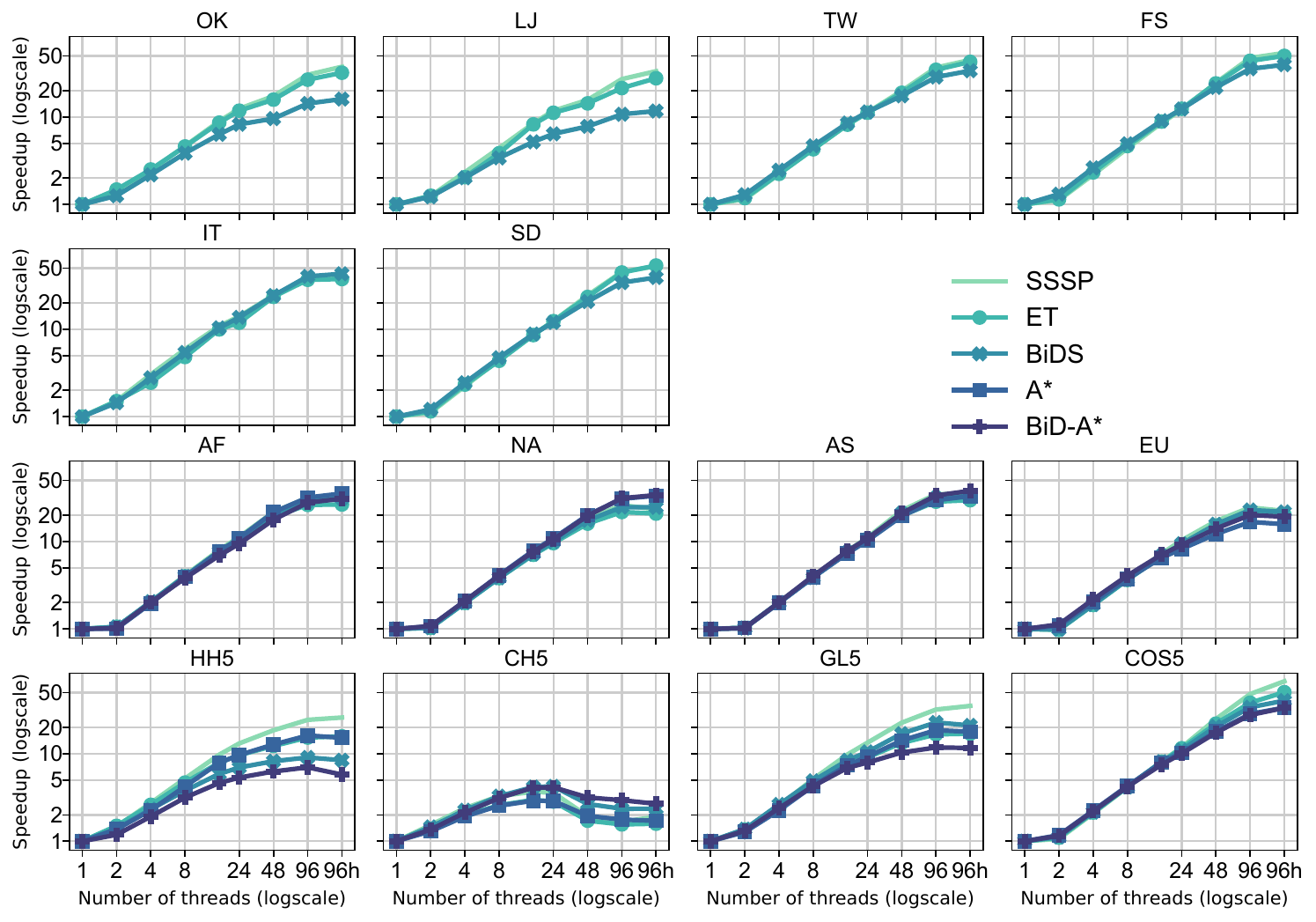}
  \caption{\textbf{Self-relative speedups of our single PPSP algorithms on each graph.}
  Higher is better.
  \label{fig:scalability-full}}
\end{figure*}

%% file: figures/memoization_full.tex
\begin{figure*}
  \small
  \centering
  \includegraphics[width=\textwidth]{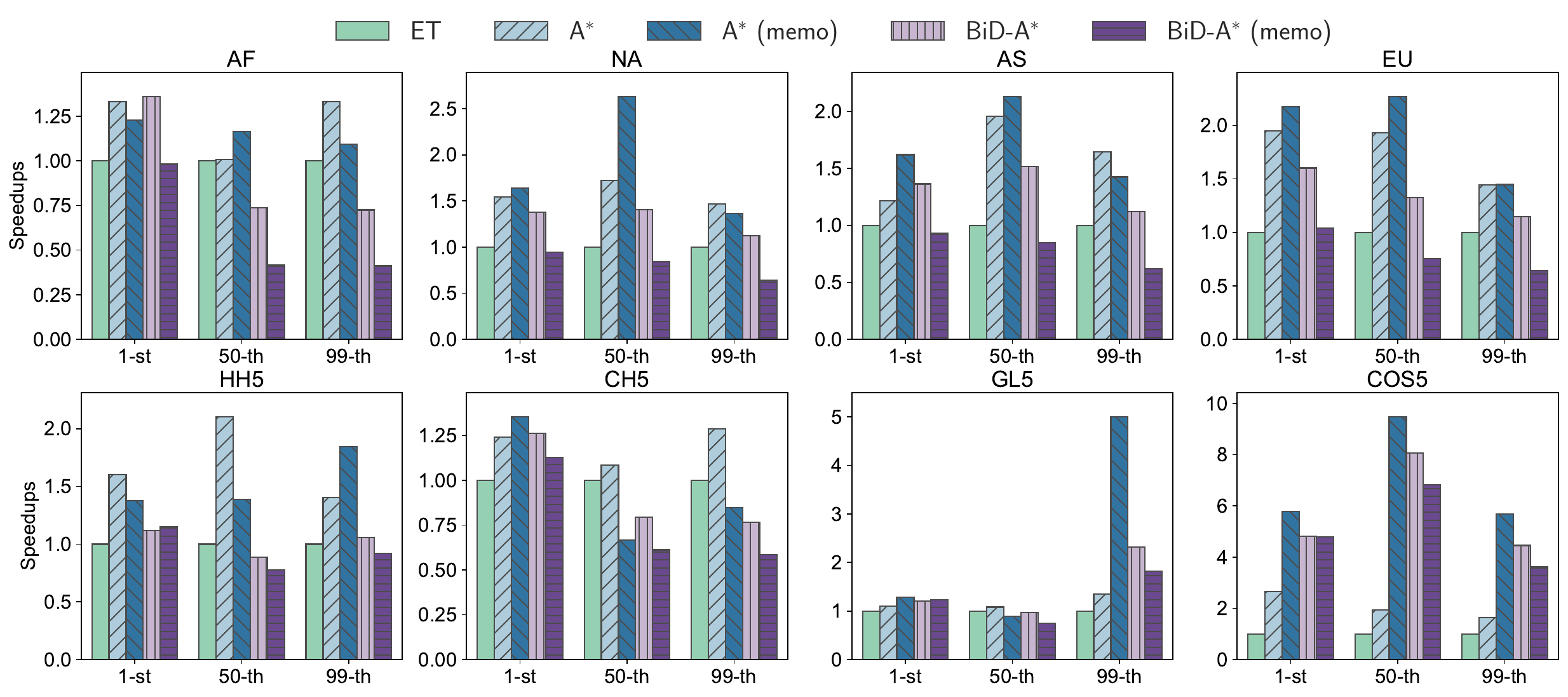}
  \caption{\textbf{Relative performance over ET with and without memoization on each graph.} 
  Numbers are running times normalized to ET. 
  ET is always one.
  Higher is better.
  \label{fig:memoization-full}}
\end{figure*}